\documentclass[a4paper,USenglish,cleveref,autoref,thm-restate]{lipics-v2021}

\usepackage{microtype}
\usepackage{xspace}
\usepackage{xcolor}
\usepackage{enumerate}
\graphicspath{./Figures/}
\title{Abstract Voronoi-like Graphs: Extending Delaunay's Theorem and
  Applications}

\author{Evanthia Papadopoulou}{Faculty of Informatics, Universit\`a
  della Svizzera italiana, {Lugano,
    Switzerland}}{evanthia.papadopoulou@usi.ch}{https://orcid.org/0000-0003-0144-7384}{}
\funding{Supported in part by
  the Swiss National Science Foundation,  project
  200021E-201356.}
      
\acknowledgements{I would like to thank Kolja Junginger for many
  preliminary  discussions, numerous figures and comments on earlier
  versions of this work.  I would also like to thank Franz Aurenhammer
  for valuable suggestions, and Elena Arseneva for constructive
  comments and figures.}

\authorrunning{E. Papadopoulou} 
\Copyright{Evanthia Papadopoulou}

\hideLIPIcs
\nolinenumbers

\EventEditors{Erin W. Chambers and Joachim Gudmundsson}
\EventNoEds{2}
\EventLongTitle{39th International Symposium on Computational Geometry (SoCG 2023)}
\EventShortTitle{SoCG 2023}
\EventAcronym{SoCG}
\EventYear{2023}
\EventDate{June 12--15, 2023}
\EventLocation{Dallas, Texas, USA}
\EventLogo{socg-logo.pdf}
\SeriesVolume{258}
\ArticleNo{54}  %

\newcommand{\deleted}[1]{}

\newcommand{\etal}{{~et~al.~}}

\newcommand{\J}{\mathcal{J}} 		
\newcommand{\V}{\mathcal{V}}

\newcommand{\VR}{\mbox{VR}}

\newcommand{\arr}{\mathcal{A}}

\newcommand{\vld}{\mathcal{V}_l}

\newcommand{\freg}{\mbox{FVR}}

\newcommand{\cdt}{\text{CDT}}

\begin{document}

\maketitle

\begin{abstract}
  Any system of bisectors (in the sense of abstract Voronoi diagrams)
  defines an arrangement of simple curves in the plane.  We define
  \emph{Voronoi-like graphs} on such an arrangement, which are graphs
  whose vertices are \emph{locally Voronoi}.  A vertex $v$ is called
  locally Voronoi, if $v$ and its incident edges appear in the Voronoi
  diagram of three sites.
  In a so-called admissible bisector system, where Voronoi regions are
  connected and cover the plane, we prove that any Voronoi-like graph
  is indeed an abstract Voronoi diagram.
  The result can be seen as an abstract dual version of Delaunay’s
  theorem on (locally) empty circles.

  Further, we define Voronoi-like cycles in an admissible bisector
  system, and show that the Voronoi-like graph induced by such a cycle
  $C$ is a unique tree (or a forest, if $C$ is unbounded). In the
  special case where $C$ is the boundary of an abstract Voronoi region,
  the induced Voronoi-like graph can be computed in expected linear time
  following the technique of [Junginger and Papadopoulou SOCG'18].
  Otherwise,  within the same time, 
  the algorithm constructs the Voronoi-like graph of a 
  cycle $C'$ on the same set (or subset) of sites,
  which may equal $C$ or be enclosed by $C$.
  Overall, the technique  computes abstract 
  Voronoi (or Voronoi-like) trees and forests
  in linear expected time, 
  given the order of their leaves along a Voronoi-like cycle.
  We show a
  direct application in updating a constraint Delaunay triangulation in
  linear expected time, after the insertion of a new segment constraint,
  simplifying upon the result of [Shewchuk and Brown CGTA 2015].



\end{abstract}


\section{Introduction}
\label{sec:intro}

Delaunay’s theorem~\cite{Del34} is a well-known 
cornerstone 
in Computational Geometry:
given a set of points, a triangulation 
is \emph{globally Delaunay} if and only if it is 
\emph{locally Delaunay}.
A triangulation edge is called \emph{locally Delaunay} if it is incident to
only one triangle, or it is incident to two triangles, and
appears in the Delaunay triangulation of the four related
vertices.
The Voronoi diagram  and the Delaunay triangulation of a point set
are dual to each other.
These two 
highly influential and  versatile structures are often 
used and computed interchangeably; see the book of
Aurenhammer\etal\cite{aurenbook} for  
extensive information. 

Let us pose the following question:
how does Delaunay's
theorem extend to Voronoi diagrams of generalized (not necessarily
point) sites?
We are interested in simple geometric objects such as line
segments,  polygons, disks, or point clusters, as they often appear
in application areas, and answering 
this question is intimately related to
efficient construction algorithms for Voronoi diagrams (or
their duals) 
on these objects.
Here we consider this question in the framework of abstract Voronoi
diagrams~\cite{K89} so that we can simultaneously answer it for various
concrete and fundamental cases under their umbrella. 

Although Voronoi diagrams and Delaunay triangulations of
point sites have been 
widely
used in many fields of science, being available in most
software libraries of commonly used programming
languages,  practice has  not been the same for their counterparts
of simple 
geometric
objects. 
In fact it is surprising that certain related questions may have remained open or
non-optimally solved. 
Edelsbrunner and Seidel~\cite{ES86} defined
Voronoi diagrams as lower envelopes of distance functions in a space one dimension
higher, making a powerful link 
to arrangements, which made their rich combinatorial and algorithmic results
applicable, e.g., ~\cite{sharir95}.
However, there are different levels of difficulty concerning
arrangements of planes versus more general surfaces,
which play a role, especially in practice.

In this paper we define Voronoi-like graphs based on 
local information, inspired by Delaunay's theorem.
Following the framework of abstract Voronoi diagrams (AVDs)~\cite{K89},
let $S$ be a set of $n$  abstract sites (a set of indices)
        and $\J$ be their  underlying system of bisectors, which satisfies
        some simple combinatorial properties (see Sections~\ref{sec:prels},~\ref{sec:vld}).
Consider a graph $G$ on the arrangement of the bisector system 
possibly truncated within a simply connected domain $D$.
The vertices of $G$ are vertices of the bisector arrangement,
its leaves lie on
the boundary $\partial D$,
and the edges are  maximal bisector arcs 
connecting pairs of vertices.
A vertex $v$ in $G$ 
is called \emph{locally Voronoi}, if $v$ and its incident edges within
a small neighborhood around $v$ appear in the Voronoi diagram
of the three sites defining $v$ (Def.~\ref{def:local-VV}), see Figure~\ref{fig:local-VV}.
%
%
The graph $G$ is called \emph{Voronoi-like}, if its vertices (other than its leaves on $\partial
D$) are locally Voronoi vertices (Def.~\ref{def:vld}), see Figure~\ref{fig:graph-g}.
If the graph $G$ is a simple cycle on the arrangement of bisectors
related to one site $p$ and its vertices are locally
Voronoi  of degree~2, then it is 
called a \emph{Voronoi-like cycle}, for brevity a \emph{site-cycle}
(Def.~\ref{def:cycle}). 

A major difference between  points in the
Euclidean plane, versus non-points, such as line segments, disks, or
AVDs, can be immediately pointed out: 
in the  former case the bisector system is a line arrangement, while
in the latter, the bisecting curves are not even pseudolines.
On a line arrangement, it is not hard to see that a Voronoi-like
graph coincides with the Voronoi diagram of the involved sites:
any Voronoi-like cycle is a convex polygon, which is in fact a
Voronoi region in the Voronoi diagram of a subset of sites.
But in the arrangement of an abstract bisector system,  many different
Voronoi-like  cycles can exist for the
same set of sites (see, e.g., Figure~\ref{fig:p-cycle-2}).
Whether a Voronoi-like graph  corresponds to a Voronoi
diagram is not immediately clear.

In this paper we show that a Voronoi-like graph on the arrangement of an abstract
bisector system is as close as possible to being an abstract Voronoi
diagram, subject to, perhaps,   
\emph{missing} some faces (see Def.~\ref{def:missing});
if the graph misses no face, then it is a Voronoi
diagram. 
 Thus, in the classic AVD model~\cite{K89}, where abstract Voronoi regions are
 connected and cover the plane, any Voronoi-like graph 
 is indeed an abstract Voronoi diagram.
This result can be seen as an abstract dual version of Delaunay’s
theorem. 

Voronoi-like graphs (and their duals) can be very useful structures to
hold partial Voronoi information, either when dealing
with disconnected Voronoi regions, or when considering 
partial information concerning some region.
Building a Voronoi-like graph of partial information may be
far easier and faster than constructing the full diagram.
In some cases the full diagram may even be undesirable, 
as in the example of Section~\ref{sec:CDT}
in updating a constrained Delaunay triangulation.
%

The term \emph{Voronoi-like diagram} was first used, in a restricted
sence, by 
Junginger and Papadopoulou~\cite{JP18}, defining it as a tree (occasionally a forest)
that subdivided a planar region enclosed by a so-called
\emph{boundary curve} 
defined on a subset of Voronoi edges.
Their Voronoi-like diagram was then used as an intermediate
structure
to perform deletion in an abstract Voronoi diagram in linear expected
time.
In this paper the formulation of a Voronoi-like graph
is entirely different;
we nevertheless prove  that the Voronoi-like diagram
of~\cite{JP18} remains 
a special case of the one 
in this paper.
We thus use the results of~\cite{JP18} 
when applicable, and extend them to Voronoi-like cycles in an
admissible bisector system.

In the remainder of this section we consider an \emph{admissible} bisector
system $\J$ following the classic AVD model~\cite{K89}, where bisectors
are unbounded simple curves and Voronoi regions are connected.
To avoid issues with infinity, we asume a large Jordan curve $\Gamma$
(e.g, a circle)
bounding the computation domain, which is large
enough to enclose any bisector intersection.
In the sequel, we list further results, which are obtained 
in this paper under this model.
%

We consider a Voronoi-like cycle $C$ on the arrangement of bisectors
$\J_p\subseteq \J \cup \Gamma$, which are \emph{related} to a site $p\in S$.
Let $S_C\subseteq S\setminus \{p\}$ be the set of sites that
(together with $p$) contribute to the bisector arcs in $C$.
The cycle $C$ encodes a \emph{sequence of site
  occurrences} from  $S_C$.
We define the Voronoi-like graph $\vld(C)$, 
which can be thought as a Voronoi diagram of \emph{site
  occurrences}, instead of \emph{sites}, whose order is represented  by $C$.
We prove that $\vld(C)$ is a tree, or a forest if $C$ is unbounded,
and  it 
exists for any Voronoi-like cycle $C$.
The uniqueness of $\vld(C)$ can be inferred from the results in
\cite{JP18}.
The same properties can be extended
to Voronoi-like graphs of cycles
related to a set $P$ of $k$ sites.

We then  consider the randomized incremental construction of
\cite{JP18},
and apply it to a Voronoi-like cycle in linear expected
time.
If $C$ is the boundary of a Voronoi region then
$\vld(C)$, which is
the part of the abstract Voronoi diagram $\V(S_C)$, truncated by $C$,
can be computed in expected linear time (this has been previously shown~\cite{JP18,JP20arxiv}).
Otherwise, within the same time, the Voronoi-like graph of a (possibly
different) Voronoi-like cycle $C’$,  enclosed by $C$, is computed by
essentially the same algorithm. 
We give conditions under which we can force the randomized algorithm
to compute $\vld(C)$, if desirable, without hurting its
expected-linear time complexity,
using deletion~\cite{JP18} as a subroutine.
The overall technique follows the randomized linear-time
paradigm of Chew~\cite{Chew90}, originally given to compute the Voronoi
diagram of points in convex position.
The generalization of Chew's technique can potentially be used to convert
algorithms working on point sites, which use it,  to counterparts
involving non-point sites that fall under the umbrella
of abstract Voronoi diagrams.



Finally, we give a direct application
for computing the Voronoi-like
graph of a site-cycle in linear expected
time, when updating a constrained Delaunay
triangulation upon insertion of a new line segment, simplifying upon the corresponding result of Shewchuk and
Brown\cite{SB15}.
The resulting algorithm is extremely simple. By modeling the problem
as computing the dual of a Voronoi-like graph, given a Voronoi-like
cycle (which is not a Voronoi region's boundary), the algorithmic description becomes almost trivial and explains
the technicalities, such as self-intersecting subpolygons, that are
listed by Shewchuk and Brown.

The overall technique computes abstract Voronoi, or 
Voronoi-like, trees and forests in linear expected time, given the
order of their leaves along a Voronoi-like cycle. 
In an extended paper, we also give simple conditions under which the
cycle $C$
is an arbitrary Jordan curve of constant complexity, along
which the ordering of Voronoi regions is known.
%

\section{Preliminaries and definitions}
\label{sec:prels}
We follow the framework of abstract Voronoi diagrams (AVDs), which have been defined by Klein~\cite{K89}.
Let  $S$ be a set of $n$  abstract sites (a set of indices)
        and $\J$ be an underlying system of bisectors that satisfy
        some simple combinatorial properties (some axioms).  
        The bisector $J(p,q)$  of two sites $p,q \in S$ is a simple curve
        that
        subdivides the plane into two open domains:
        the \emph{dominance region of $p$}, $D(p,q)$,
	having label $p$,
	and the \emph{dominance region of $q$}, $D(q,p)$,
	having label $q$. 
%

        The \emph{Voronoi region} of site $p$ is 
        
	\[\VR(p,S) = \bigcap_{q \in S \setminus \{p\}}  D(p,q).\]
	\noindent
	The \emph{Voronoi diagram} of $S$ is
	$\V(S) = \mathbb{R}^2\setminus \bigcup_{p \in S}\VR(p, S)$.
        The vertices and the edges of $\V(S)$ are called \emph{Voronoi vertices}
        and 
        \emph{Voronoi edges}, respectively.

        Variants of abstract Voronoi diagrams
        of different degrees of generalization
        have been
        proposed, see e.g., \cite{KLN09, BKL17}. 
 %
        Following the original formulation by Klein~\cite{K89}, the bisector system  $\J$  is called  \emph{admissible}, if it satisfies the following axioms,
	for every subset $S' \subseteq {S}$:

	\begin{enumerate}
		\item[(A1)] Each Voronoi region $\VR(p, S')$ is
                  non-empty and pathwise connected. 
		\item[(A2)] Each point in the plane belongs to the
                  closure of a Voronoi region $\VR(p, S')$.
                  \item[(A3)] Each bisector is an unbounded simple
                    curve homeomorphic to a line.
		\item[(A4)] Any two bisectors intersect
                  transversally and in a finite number of points.

                \end{enumerate}

Under these axioms, the abstract Voronoi diagram $\V(S)$ is a planar graph
of complexity $O(n)$, which can be computed in 
$O(n\log n)$  time, randomized
\cite{KMM93} or  deterministic~\cite{K89}.

       To avoid dealing with infinity,  we
       assume that $\V(S)$ is truncated within a domain $D_\Gamma$
       enclosed by a large Jordan curve
       $\Gamma$ (e.g., a circle or a rectangle)
       such that all bisector intersections are contained in $D_\Gamma$.
	Each bisector crosses $\Gamma$ exactly  twice and
        transversally.
        All Voronoi regions are assumed 
        to be  truncated by $\Gamma$, and thus, lie  within the domain
        $D_\Gamma$.

        We make a general position assumption that no three 
        bisectors involving one common site intersect at the same point,
        that is, all vertices in the arrangement of the 
        bisector system $\J$ have degree 6, and Voronoi vertices
        have degree 3.

  \begin{figure}
		\centering
		\includegraphics[scale=0.25]{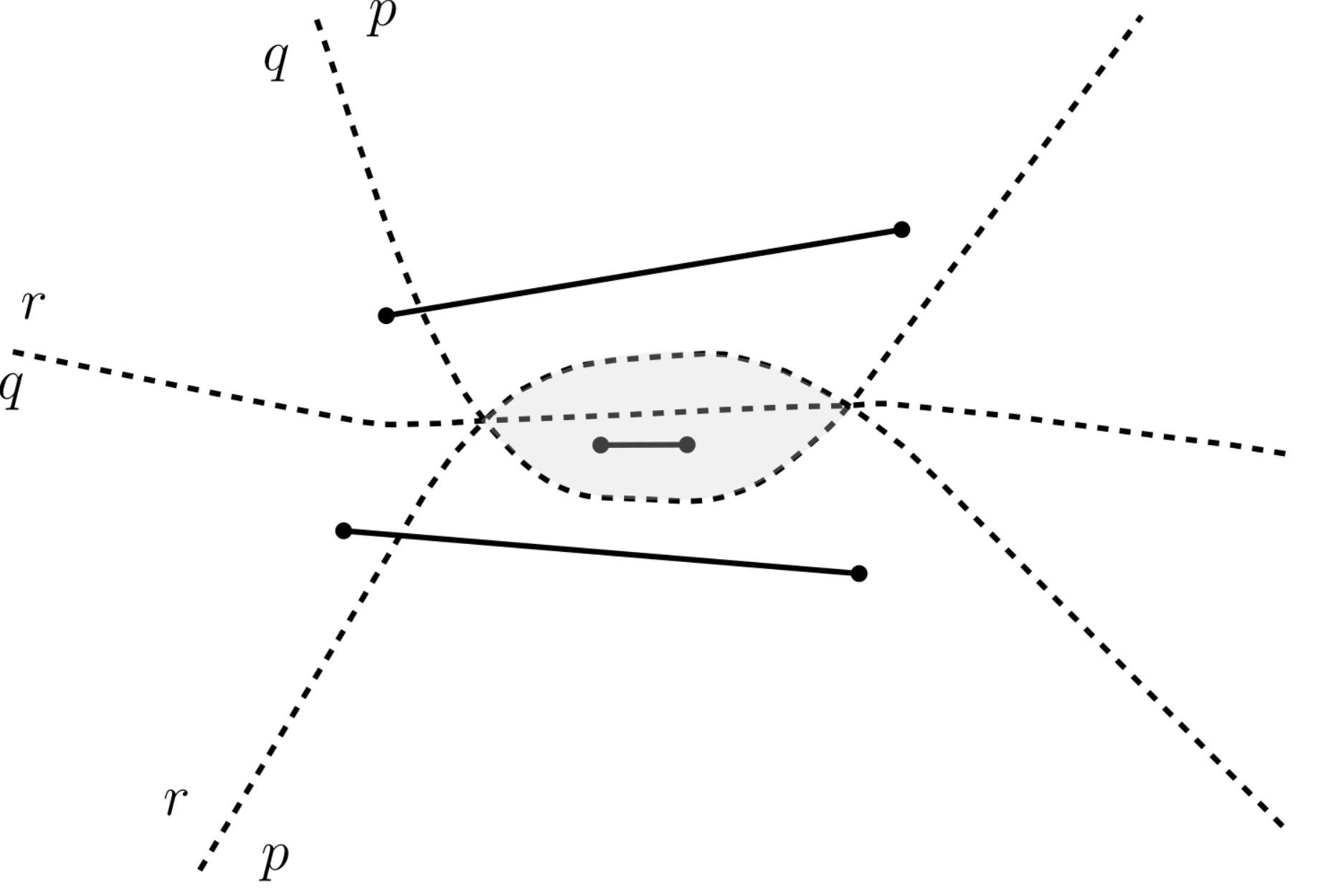}
                \vspace*{0.5\baselineskip}
		\caption{
			Related segment bisectors intersecting
                        twice. $\VR(p,\{p,q,r\})$ is shaded. 
                      }
		\label{fig:3-segment-bisectors}
\end{figure}        

        Bisectors that have a site $p$ in common are called
        \emph{related}, in particular, \emph{$p$-related}.
        Let $\J_p\subseteq \J$ denote the set of all $p$-related
        bisectors in $\J$.
        Under axiom A2, if related bisectors $J(p,q)$ and $J(p,s)$ intersect 
        at a vertex $v$, then $J(q,s)$ must also intersect with them
        at the same vertex,
        which is  a Voronoi vertex in $V(\{p,q,s\})$
        (otherwise, axiom A2 would be violated in $V(\{p,q,s\})$). 
        %
        In an admissible bisector system, related bisectors can
        intersect at most twice~\cite{K89};
        thus, 
        a Voronoi diagram of three sites may have at most  two Voronoi
        vertices, see e.g., the bisectors of three line segments in
        Figure~\ref{fig:3-segment-bisectors}.
        The curve $\Gamma$ can be interpreted as a $p$-related bisector
        $J(p,s_\infty)$, for a site $s_\infty$ representing
        infinity, for any $p\in S$.

\begin{observation}
\label{obs:patterns}
In an admissible bisector system, related bisectors that do not 
intersect or intersect twice must follow the patterns illustrated in Figures~\ref{fig:bisector-0-intersection}
and~\ref{fig:bisector-2-intersection} respectively.
\end{observation}

\begin{proof}
In Figure~\ref{fig:bisector-2-intersection}(c) the pattern is illegal
because of axiom A1, and in Figure~\ref{fig:bisector-2-intersection}(d) because
of combining axioms A2 and A1:
$J(s,t)$ must pass through the intersection
points of $J(p,s)$ and $J(t,p)$, by A2. 
Then any possible configuration of $J(s,t)$ 
results in violating either axiom A1 or A2. 
The pattern in Figure~\ref{fig:bisector-0-intersection}(b) can
be shown illegal by combining axioms A1 and A2 in the
presence of  $J(s,t)$, which does not intersect $J(s,p)$ nor $J(t,p)$. 
\end{proof}


\begin{figure}
    \centering
    \includegraphics[scale=0.9]{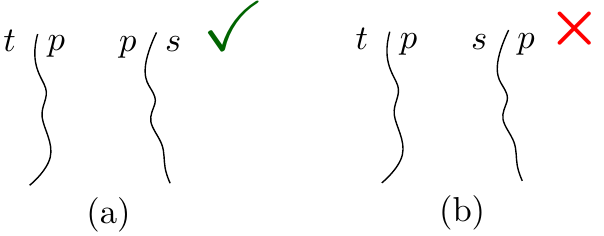}
    \vspace*{0.5\baselineskip}
    \caption{Non-intersecting bisectors; (a)  
      is legal  ($\checkmark$); 
      (b) is illegal ($\times$).} 
    \label{fig:bisector-0-intersection}
  \end{figure}
  
\begin{figure}
		\centering
		\includegraphics[scale=0.9]{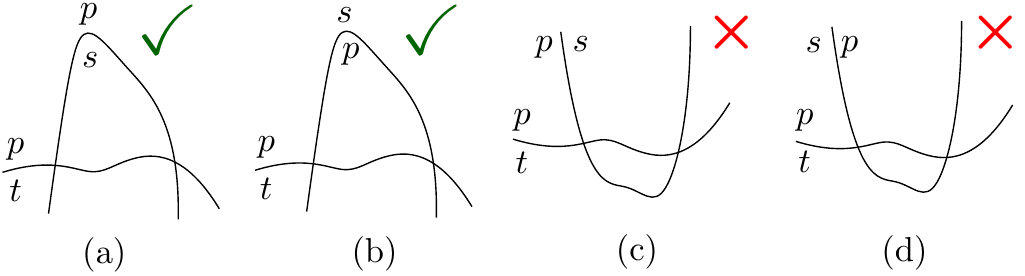}
                \caption{Bisectors intersecting twice; legal
                ($\checkmark$) and illegal($\times$).}
              \label{fig:bisector-2-intersection}
       \end{figure}

\begin{observation} [\cite{JP18}]
\label{obs:inversecycle}
In an admissible bisector system, no cycle in the arrangment of
bisectors related to $p$
can have the label $p$ on the exterior of the cycle, for all of its
arcs.
\end{observation}

        Any component $\alpha$ of a bisector curve $J(p,q)$ is called an
        \emph{arc}. We use $s_\alpha \in S$ to denote the site such  that
         arc $\alpha \subseteq J(p,s_\alpha)$.
        Any component of $\Gamma$ is called a $\Gamma$-arc.
        The arrangement of a bisector set $\J_x\subseteq \J$ is denoted by
        $\arr(\J_x)$.

\section{Defining abstract Voronoi-like graphs and cycles}
\label{sec:vld}

In order to define Voronoi-like graphs in a
broader sense, we can relax  axioms A1-A4 in this section.
In particular, we  drop axiom A1 to allow disconnected
Voronoi regions and relax axiom A3 to allow disconnected (or even closed)
bisecting curves. 
The bisector $J(p,q)$  of two sites $p,q \in S$
        still subdivides the plane into two  open domains:
        the \emph{dominance region of $p$}, $D(p,q)$,
	and the \emph{dominance region of $q$}, $D(q,p)$,
        however, $D(p,q)$ 
        may be disconnected or  bounded.
        Axioms A2 and A4 remain.
        Unless otherwise specified, we use the general term \emph{abstract bisector
          system} to denote such a relaxed variant in
        the subsequent definitions and in
        Theorem~\ref{thrm:face}.
        The term \emph{admissible bisector system} always implies
        axioms  A1-A4.
 

Let $G=(V,E)$ be a graph on the arrangement of an abstract bisector system $\J$,
truncated within a simply connected domain $D\subseteq D_\Gamma$ (the
leaves of $G$ are on  $\partial D$).
The vertices of $G$ are  arrangement vertices 
and the edges are maximal bisector arcs connecting pairs of vertices.
Figure~\ref{fig:graph-g} illustrates examples of such graphs on a
bisector arrangment (shown in grey).
Under  the general position assumption, the vertices of $G$, except
the leaves on $\partial D$,  are of degree 3.

\begin{definition}\label{def:local-VV}
A vertex $v$ in graph $G$ 
is called \emph{locally Voronoi}, if $v$ and its incident graph edges, within
a small neighborhood around $v$, $N(v)$, appear in the Voronoi diagram
of the set of three sites defining $v$, denoted $S_v$, 
see Figure~\ref{fig:local-VV}(a).

If instead we consider the farthest
Voronoi diagram of 
$S_v$, then $v$ is called locally
Voronoi of the \emph{farthest-type}, see Figure~\ref{fig:local-VV}(b).
An ordinary locally Voronoi vertex  is 
of the \emph{nearest-type}.
\end{definition}


        	\begin{figure}
		\centering
		\includegraphics[scale=0.85]{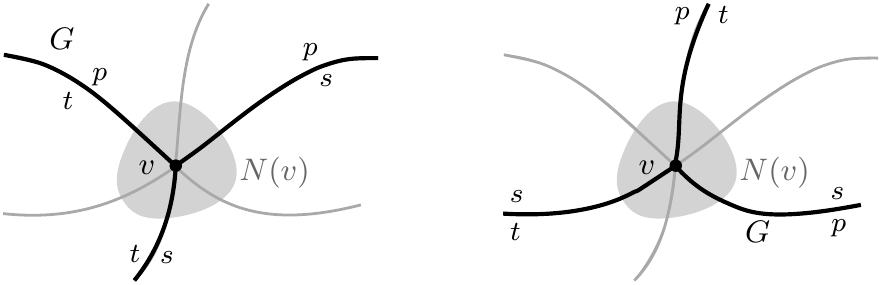}
		\caption{(a)  Vertex $v$ 
                  is locally Voronoi:
			$G \cap N(v) = 
			\V(\{p,s,t\})\cap N(v)$;
                        $N(v)$ is shaded, $G$ is  bold, and bisectors are
                        grey.
                        (b)  Vertex $v$ 
                        is locally farthest Voronoi.
                        	}
		\label{fig:local-VV}
              \end{figure}

\begin{definition}
\label{def:vld}
A graph $G$ on the arrangement of an abstract bisector system,
enclosed within a simply connected domain $D$, is called \emph{Voronoi-like}, if its vertices (other than its leaves on $\partial
D$) are locally Voronoi vertices.
If $G$ is disconnected, we further require that consecutive leaves
on $\partial D$ have \emph{consistent labels}, i.e., they are incident 
to the dominance region of the same site, as implied  by the incident
bisector edges in $G$, see Figure~\ref{fig:graph-g}.
%
\end{definition}

\begin{figure}
  \centering
  \includegraphics[scale=0.85]{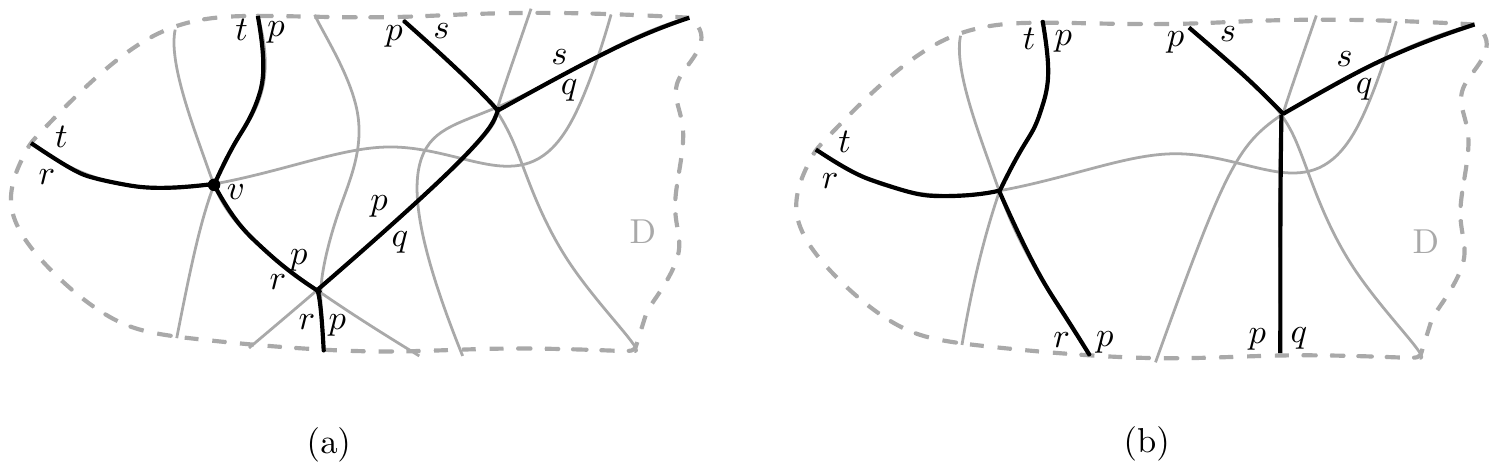}
  \caption{Voronoi-like graphs shown in bold on an
    arrangement of bisectors shown in grey.		}
  \label{fig:graph-g}
\end{figure}

The graph  $G$ is actually called an \emph{abstract Voronoi-like graph} but, 
for brevity, we typically skip the term abstract.
We next consider the relation between a Voronoi-like graph $G$ and the Voronoi diagram
$\V(S)\cap D$, where $S$ is the set of sites involved in the edges of  $G$. 
Since the vertices of $G$ are locally Voronoi, each face $f$ in $G$
must have the label of exactly one site $s_f$ in its interior, which
is called the \emph{site of $f$}.

\begin{definition}
  \label{def:missing}
Imagine we superimpose $G$ and  $\V(S)\cap D$. 
A face $f$ of $\V(S)\cap D$ is said to be \emph{missing from $G$}, if
 $f$ is covered by faces of $G$ that belong to sites that
 are different
 from the site of $f$, see Figure~\ref{fig:voronoi-missing-face}, which
 is derived from Figure~\ref{fig:voronoi-segments}.
\end{definition}

\begin{figure}
     \centering
    \includegraphics[scale=0.8]{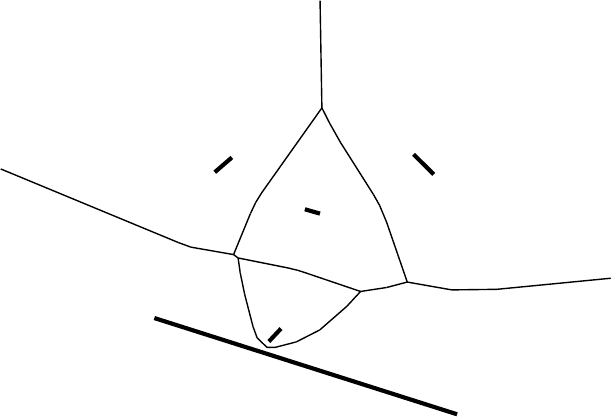}
    \caption{A Voronoi diagram of 4 segments.
    }
    \label{fig:voronoi-segments}
  \end{figure}

  \begin{figure}
 		\centering
		\includegraphics[scale=0.8]{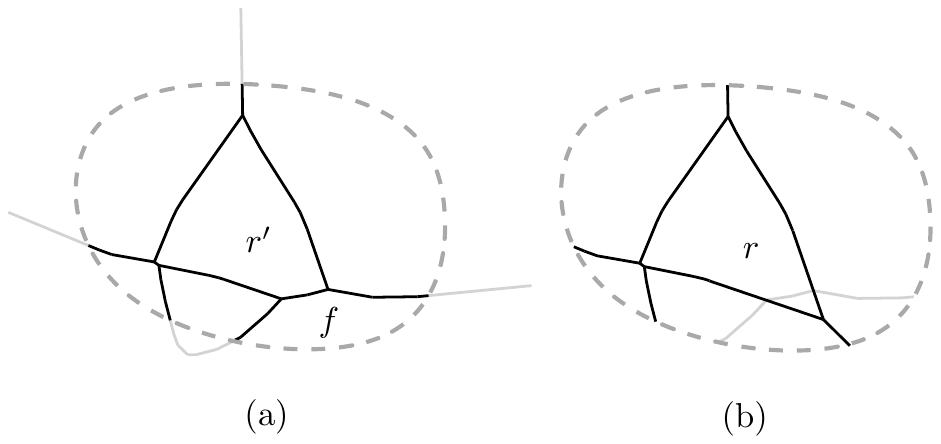}
		\caption{(a) Voronoi diagram  $V(S)\cap D$; (b)
                  Voronoi-like graph $G$; 
                  $r'\subseteq r$; $f$ is missing from (b).}
		\label{fig:voronoi-missing-face}
      \end{figure}

\begin{theorem}
  \label{thrm:face}
Let $r$ be a face
of an abstract Voronoi-like graph $G$ and let  $s_r$  denote its site (the bisectors
bounding $r$ have the label $s_r$ inside $r$). Then one of the
following holds:
\begin{enumerate}
\item there is a Voronoi face $r'$ in $\V(S)\cap D$, of the same site
  $s_r$, $r'\subseteq \VR(s_r,S)$, such that $r'\subseteq r$, see
  Figure~\ref{fig:voronoi-missing-face}.
\item 
face $r$ is disjoint from the Voronoi region $\VR(s_r,S)$. 
Further, it is
entirely covered by Voronoi faces of $\V(S)\cap D$, which are
\emph{missing} from $G$, see
  Figure~\ref{fig:face-2}.
\end{enumerate}
\end{theorem}

\begin{proof}
Imagine we superimpose $G$ and  $\V(S)\cap D$.
Face $r$ in $G$ cannot partially overlap any face of the Voronoi
region $\VR(s_r,S)$
because if it did, some $s_r$-related bisector, 
which contributes to the boundary of $r$, would
intersect the interior of  $\VR(s_r,S)$, which is not possible by the
definition of a Voronoi region. For the same reason, $r$ cannot be
contained in $\VR(s_r,S)$. Since Voronoi regions cover the plane, the
claim, except from the last sentence in item 2, follows.

Consider a Voronoi face $c'$ of $\V(S)\cap D$ that  overlaps with 
face $r$ of $G$ in case 2, where the site of $c'$, $s_c$, is different from $s_r$.
Since $c'$ overlaps with $r$, it follows that $c'$ cannot be entirely  contained
in any face of site $s_c$ in $G$.
Furthermore,  $c'$ cannot overlap partially  with any face of
$s_c$ in  $G$, by the proof in the previous paragraph. 
Thus, $c'$ is disjoint from any face of  $G$ of site $s_c$, i.e.,
it must be missing from $G$.
In Figure~\ref{fig:face-2}, face $c'$ contains $r$.
\end{proof}

 \begin{figure}
 		\centering
		\includegraphics[scale=0.8]{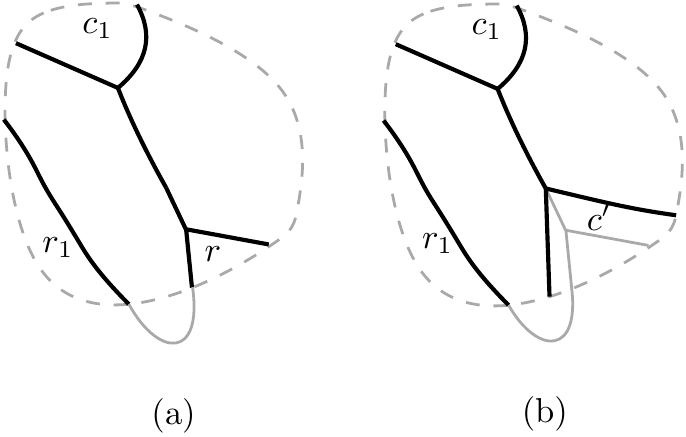}
		\caption{(a)
                  Voronoi-like graph $G$; 
                  (b)  $V(S)\cap D$; face $c'$, which
                  covers $r$,  is missing
                  from $G$.}
		\label{fig:face-2}
      \end{figure}

      
\begin{corollary}
\label{cor:no-missing-face}
If no Voronoi face of  $\V(S) \cap D$ is missing from $G$, then $G = \V(S) \cap D$.
\end{corollary}

Let us now consider an  admissible bisector system, satisfying axioms
A1-A4. 

\begin{corollary}
\label{cor:AVD-global}
In an \emph{admissible bisector system $\J$}, if $D$ corresponds to the entire plane,
then any Voronoi-like graph  on $\J$ equals the Voronoi
diagram of the relevant set of sites. 
\end{corollary}


In an admissible bisector system, Voronoi regions are connected, thus,
only faces incident to $\partial D$ may be missing from  $\V(S) \cap D$.

\begin{corollary}
\label{cor:AVD-D}
In an admissible bisector system, any  
face $f$ of $G$ that does not touch $\partial D$ either
coincides with or contains the
Voronoi region $\VR(s_f,S)$. 
\end{corollary}

By Corollary~\ref{cor:AVD-D},
in an admissible bisector system, we
need to characterize 
the faces of a Voronoi-like graph
that interact with the boundary of the domain $D$. 
That is, we are interested in Voronoi-like
trees and forests.

Let $p$ be a site in $S$ and let $\J_p$  denote the set of $p$-related
bisectors in $\J$.

     \begin{definition}
        \label{def:cycle}
        Let $C$ be a cycle  in the arrangement of $p$-related bisectors $\arr(\J_p\cup
        \Gamma)$ such that the label $p$ appears in the interior of $C$. 
        A vertex $v$ in $C$ is called  \emph{degree-2
          locally Voronoi}, if its two incident bisector arcs
        correspond  to edges
        in the Voronoi diagram $\V(S_v)$
        of the three sites that define $v$ 
        ($p\in S_v$).
        In particular, $C\cap N(v)\subseteq \V(S_v)\cap N(v)$, where
        $N(v)$ is a small neighborhood around $v$.
        The cycle $C$ is called \emph{Voronoi-like}, if its vertices
        are either degree-2 locally Voronoi or points on $\Gamma$.
        For brevity, $C$ is also called  a 
        \emph{$p$-cycle} or \emph{site-cycle}, if the site $p$ is not specified.
        If $C$ bounds a Voronoi region, then it is called a \emph{Voronoi cycle}. 

  %
      $C$ is called  \emph{bounded} if it contains no $\Gamma$-arcs, otherwise, it is called
      \emph{unbounded}.
      \end{definition}


      The part of the plane enclosed by $C$ is called the
        \emph{domain of $C$}, denoted as $D_C$.
        Any $\Gamma$-arc of $C$ indicates an opening of the 
        domain to infinity.
%
       Figure~\ref{fig:p-cycle} illustrates a Voronoi-like cycle for
       site $p$,
       which is unbounded (see the $\Gamma$-arc $\gamma$).
        It is easy to see in this
        figure that other $p$-cycles exist, on the same set of sites,
        which may enclose or 
        be enclosed by $C$. The innermost such cycle is the boundary
        of a Voronoi region, see Figure~\ref{fig:p-cycle-2}. 

       \begin{figure}
 		\centering
		\includegraphics[scale=0.75]{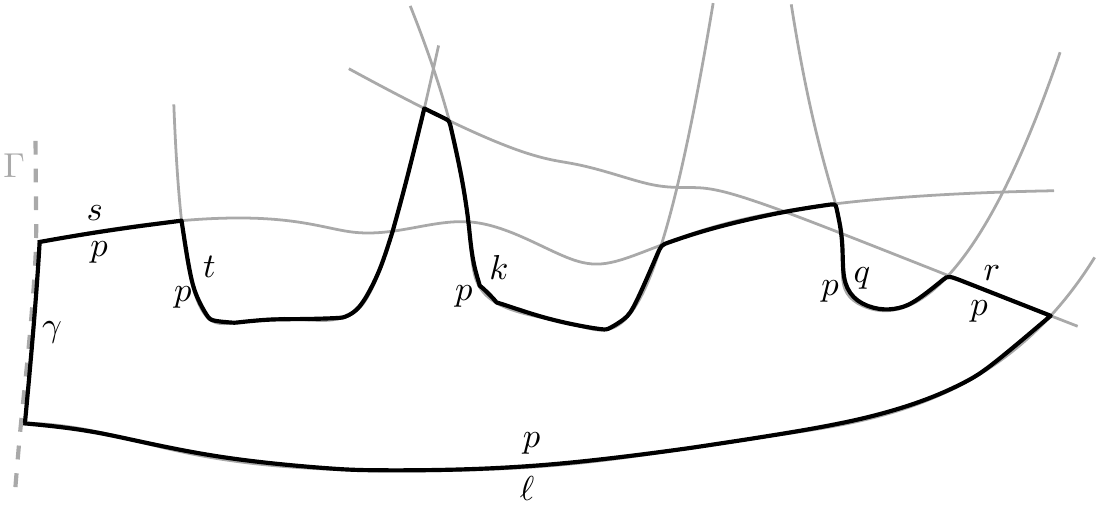}
		\caption{ A
                  Voronoi-like cycle for site $p$. $S_c=\{s,t,k,q,r,l\}$.}
		\label{fig:p-cycle}
      \end{figure}

        Let $S_C\subseteq S\setminus \{p\}$ denote the set of sites that
        (together with $p$) contribute the bisector arcs of $C$, 
        $S_C=\{s_\alpha\in S\setminus \{p\}\mid  \alpha \in
        C\setminus\Gamma\}$.
        We refer to $S_C$ as the  \emph{set of sites relevant} to $C$.
        Let $\hat C$ denote the Voronoi cycle $\hat
        C=\partial(\VR(p,S_C\cup\{p\})\cap D_\Gamma)$.


%

        \begin{observation}
          \label{obs:encloseVR}
          In an admissible bisector system,
          there can be many
        different Voronoi-like cycles  involving the
        same set of sites. Any 
        such cycle $C$ must enclose  the Voronoi cycle $\hat C$.
        Further, $S_{\hat C}\subseteq S_C$.
       %
        In the special case of a  line
        arrangement, e.g., bisectors of point-sites in
        the Euclidean plane,
        a site-cycle $C$ is unique for $S_C$; in particular,
        $C=\hat C$. 
      \end{observation}

A Voronoi-like cycle $C$ must share several bisector arcs with its
Voronoi cycle $\hat C$, at least one bisector arc
for each site in $S_{\hat C}$.
Let $C\cap\hat C$ denote the sequence of common arcs between $C$ and $\hat C$.
 %
%
       Several other $p$-cycles $C'$, 
       where $S_{\hat
        C}\subseteq S_{C'}\subseteq S_C$, 
      may lie between  $C$ and  $\hat C$, all sharing  
      $C\cap\hat C$.
      Other $p$-cycles 
      may enclose $C$.
      Figure~\ref{fig:p-cycle-2} shows such cycles, where the innermost one is
      $\hat C$; its domain (a Voronoi region) is shown in solid grey.
      
            \begin{figure}
 		\centering
		\includegraphics[scale=0.75]{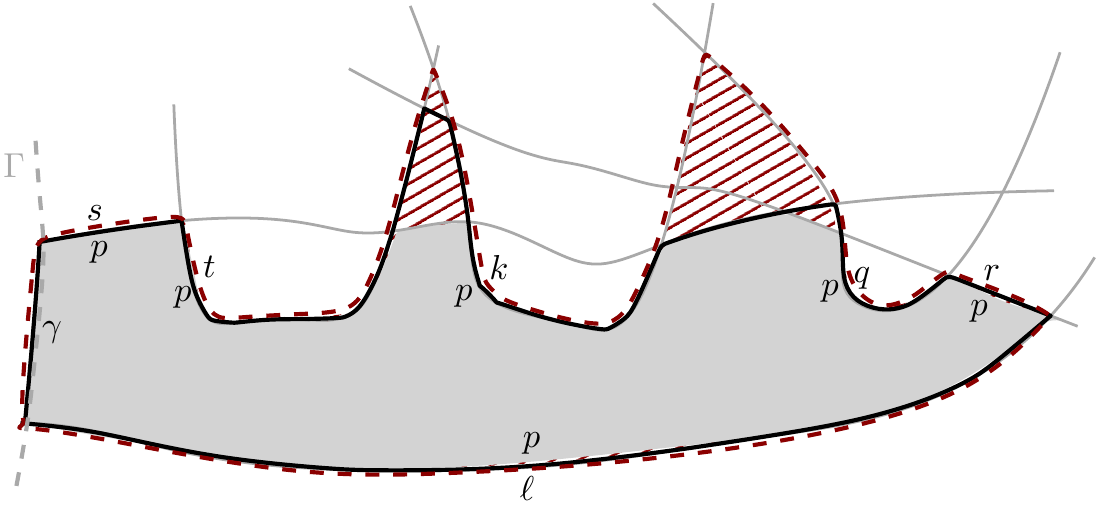}
		\caption{
                  Voronoi-like cycles for site $p$,
                  $S_c=\{s,t,k,q,r,l\}$.}
		\label{fig:p-cycle-2}
      \end{figure}

\section{The Voronoi-like graph of a cycle}
\label{sec:cycle}
Let $\J$ be an admissible bisector system and let $C$ be a
Voronoi-like cycle 
for site $p$, which involves a set of sites $S_C$ ($p\not\in
S_C$).
Let $\J_{C}\subseteq \J$ be the subset of all bisectors that are related
to the sites in $S_c$.
The cycle
$C$ corresponds to a sequence of \emph{site-occurrences} from $S_C$, which  imply
a Voronoi-like graph $\vld(C)$ in the domain of $C$, defined as
follows:

\begin{definition}
\label{def:vldC}
The Voronoi-like graph $\vld(C)$, implied by a Voronoi-like cycle $C$, is a graph on
the underlying arrangement of bisectors $\arr(\J_C)\cap D_C$,  
whose leaves are the vertices of $C$, and its remaining (non-leaf) vertices are locally
Voronoi vertices, see Figure~\ref{fig:vld}.

\noindent
(The existence of such a graph
on $\arr(\J_C)\cap D_C$ remains to be established).
\end{definition}

       \begin{figure}
 		\centering
		\includegraphics[scale=0.75]{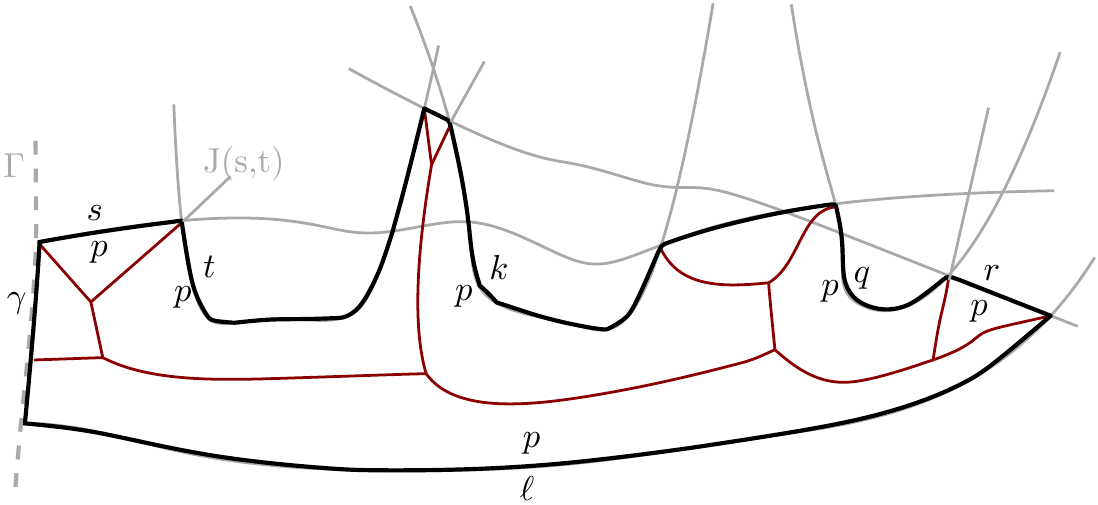}
		\caption{The 
                  Voronoi-like graph $\vld(C)$ (red tree) of the site-cycle $C$ of Fig.~\ref{fig:p-cycle}}
		\label{fig:vld}
      \end{figure}

In this section we prove the following theorem for
any Voronoi-like cycle $C$ on $\arr(\J_p\cup\Gamma)$.
%

\begin{theorem}
\label{thrm:vld}
The Voronoi-like graph $\vld(C)$ of a $p$-cycle $C$ has the following
properties:
\begin{enumerate}
\item it exists and is unique;
  \item it is a tree if $C$ is bounded, and a forest if $C$ is
    unbounded;
  \item it can be computed in expected linear time, if it is the
    boundary of a Voronoi region. 
    Otherwise, in expected linear time we can compute $\vld(C')$
     for some $p$-cycle $C'$ 
     that is enclosed by $C$ (possibly, $C'=C$ or $C'=\hat C$).
\end{enumerate}
\end{theorem}

Recall that $\hat C$ denotes the Voronoi-cycle enclosed by $C$, where
$\hat C=\partial [\VR(p,S_C\cup\{p\})\cap D_\Gamma]$.
Then $\vld(\hat C)$  is the Voronoi
diagram $\V(S_C)\cap D_{\hat C} $.
To derive Theorem~\ref{thrm:vld} we show each item separately in
subsequent lemmas.

\begin{lemma}
\label{lem:vld}
Assuming that it exists, $\vld(C)$ is a forest, and if $C$ is bounded,
then $\vld(C)$ is a tree.
Each face of $\vld(C)$ is incident to exactly one bisector arc
$\alpha$ of
$C$, which is called the face (or region) of $\alpha$,
denoted $R(\alpha,C)$.
\end{lemma}

\begin{proof}
  We first show that $\vld(C)$ contains no cycles.  
By Observation~\ref{obs:encloseVR}, any
Voronoi-like cycle for a site $s\in S_C$  must entirely enclose 
$\VR(s,S_C)$, thus, it must also enclose
$\VR(s,S_C\cup\{p\})\subseteq \VR(s,S_C)$.
Since  $J(p,s)$ 
contributes arc(s) to $C$,
it follows that $\VR(s,S_C\cup\{p\})$ must extend outside of $C$,
hense, $\VR(s,S_C)$ must also extend outside of $C$. 
Since $\VR(s,S_C)$ cannot be enclosed by $C$, the same must hold for any
$s$-cycle on $S_C$. 
Thus, $\vld(C)$ may not contain a cycle. 
The same argument  implies that 
  $\vld(C)$ cannot have a face that is incident to $\Gamma$ without
also  being  incident to a bisector arc of $C$. 

       \begin{figure}
 		\centering
		\includegraphics[scale=0.95]{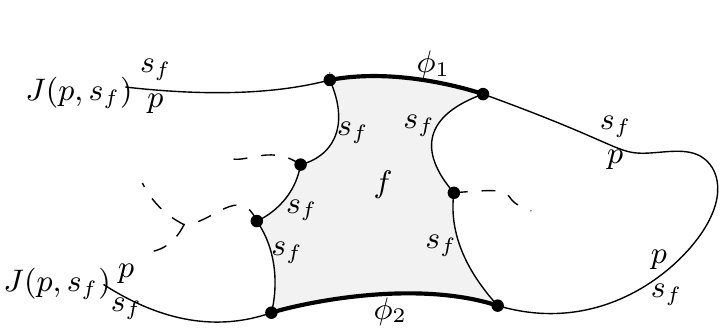}
		\caption{Proof of Lemma~\ref{lem:vld}; face $f$ is
                  shown shaded. }
		\label{fig:lemma:vld}

      \end{figure}

  Suppose now that $\vld(C)$ has a face $f$, which 
  belongs to a site $s_f\in S_c$, incident to two bisector arcs 
  $\phi_1,\phi_2\in C$ such that  $\phi_1,\phi_2 \subseteq J(p,s_f)$,
   see Figure~\ref{fig:lemma:vld}.
  Then one brunch of $\partial
          f\setminus\{\phi_1,\phi_2\}$ and the component of
          $J(p,s_f)$ between $\phi_1$ and $\phi_2$ would form a cycle
          having  the label $s_f$ 
          outside, see Figure~\ref{fig:lemma:vld}.
          Such
          a cycle is not possible in an admissible bisector system, 
          by Observation~\ref{obs:inversecycle}, deriving a contradiction.
          Thus, each face of $\vld(C)$ must be incident to exactly one
          bisector arc.
\end{proof}

If $C$ is the boundary of a Voronoi region, the tree property of the
Voronoi diagram  $\V(S)\cap D_C$ had been previously shown in~\cite{JP18,BKL19}.
Lemma~\ref{lem:vld} generalizes it to Voronoi-like
graphs for  any Voronoi-like cycle $C$. 

In~\cite{JP18}, a \emph{Voronoi-like diagram} was defined 
as a tree structure subdividing the domain of a so-called
\emph{boundary curve}, which was implied by a set of Voronoi edges.  A
boundary curve is a Voronoi-like cycle but not necessarily vice versa. 
That is, the tree structure of~\cite{JP18} was defined using 
some of the properties in Lemma~\ref{lem:vld} as definition,
and the question whether such a tree always existed
had remained open.
In this paper a Voronoi-like graph is defined entirely differently, but
Lemma~\ref{lem:vld} implies that 
the two structures are equivalent within the domain of a boundary
curve. As a result,
we can use and extend the results of \cite{JP18}. 

Given a $p$-cycle $C$,
and a bisector $J(p,s)$ that  intersects it,
an \emph{arc-insertion operation} can be defined~\cite{JP18}  as follows.
Let $\alpha\subseteq J(p,s)$ be a maximal
component of $J(p,s)$ in the domain of $C$, see Figure~\ref{fig:cases}.
         Let $C_\alpha=C\oplus
         \alpha$  denote the $p$-cycle 
         obtained by substituting with $\alpha$ the superflous portion of $C$ between the
         endpoints of $\alpha$. 
         (Note that only one portion of $C$
         forms a $p$-cycle with $\alpha$, thus, no ambiguity exists).
          There are three different main cases possible as a result, see Figure~\ref{fig:cases}:
         1) $\alpha$ may lie between two consecutive
         arcs of $C$, in which case  $|C_\alpha|=|C|+1$; 2) $\alpha$ may cause the deletion of 
         one or more arcs in $C$, thus, $|C_\alpha|\leq |C|$; 3) the endpoints of
         $\alpha$ may lie on the same arc $\omega$ of $C$, in which
         case $\omega$ splits in two different arcs, thus,
         $|C_\alpha|=|C|+2$. 
         In all cases $C_\alpha$ is enclosed by $C$ ($|\cdot|$
         denotes cardinality).
         
       \begin{figure}
 		\centering
		\includegraphics[scale=0.95]{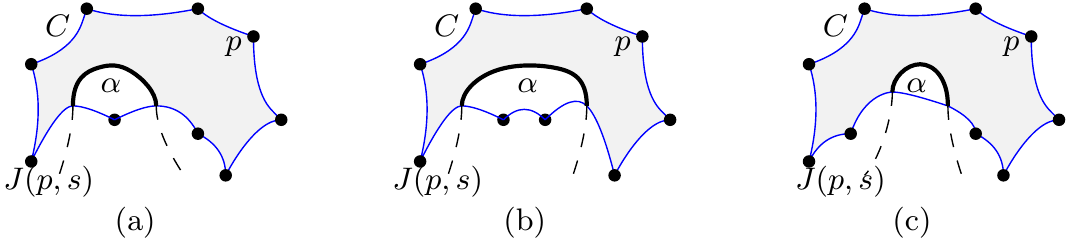}
		\caption{Three cases of the arc insertion operation.}
		\label{fig:cases}
               
              \end{figure}
              
The arc-insertion operation can be naturally extended  to the Voronoi-like
graph $\vld(C)$ to insert arc $\alpha$ and obtain
$\vld(C_\alpha)$. We use
the following lemma, which can be extracted from~\cite{JP18}
(using Theorem~18, Theorem~20, and Lemma~21 of \cite{JP18}). 


\begin{lemma}[\cite{JP18}]
\label{lem:insertion}
Given $\vld(C)$, arc $\alpha\in J(p,s) \cap D_C$, and the endpoints of $\alpha$ on $C$, we can
compute the merge curve 
$J(\alpha)=\partial R(\alpha, C_\alpha)$, using standard
techniques 
as in ordinary Voronoi diagrams.
If the endpoints of $\alpha$ lie on different arcs of $C$, or $\Gamma$,
the time complexity is  $O(|J(\alpha)| +|C\setminus C_\alpha|)$.
Otherwise,  $\alpha$ \emph{splits} a bisector arc $\omega$, and
its region  $R(\omega,C)$, into
$R(\omega_1,C_\alpha)$ and $R(\omega_2,C_\alpha)$;
the time complexity increases to  $O(|J(\alpha)|+
\min\{|R(\omega_1,C_\alpha)|, |R(\omega_2,C_\alpha)|\})$.
\end{lemma}

The correctness proofs from \cite{JP18,JP20arxiv}, which are related
to Lemma~\ref{lem:insertion},  
remain intact if performed on a Voronoi-like cycle, as long as the arc $\alpha$ is
contained in the cycle's domain; see also 
\cite[Lemma~9]{JP20arxiv}.
Thus, Lemma~\ref{lem:insertion} can be established.

Next we prove the existence of $\vld(C)$ by construction.
To this goal we use a \emph{split relation} between bisectors in
$\J_p$ or sites in $S_C$, which had also been considered in
\cite{JP19}, see Figure~\ref{fig:split}. 

\begin{definition}\label{def:split}
For any two sites $s,t\in S_C$, 
we say that $J(p,s)$ \emph{splits}
$J(p,t)$  (we also say that $s$ splits $t$, with respect
to $p$),
if $J(p,t) \cap D(p,s)$ 
contains two connected components.
\end{definition}

\begin{figure}
 	\centering 
 		
 		\includegraphics{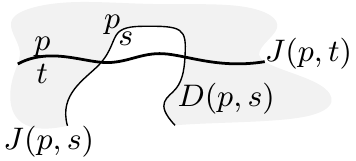}
 		\caption{$J(p,s)$ splits  $J(p,t)$. In the partial order, $s<_pt$.}
 		\label{fig:split}
\end{figure}
 
From the fact that related bisectors in an admissible bisector system
intersect at most twice, as shown in
Figs.~\ref{fig:bisector-0-intersection} and \ref{fig:bisector-2-intersection},
we can infer
that the split relation  is asymmetric and transitive, 
thus, it is also acyclic.
The split relation  induces a strict partial order on $S_C$,
where $s<_p t$, if $J(p,s)$ \emph{splits} $J(p,t)$, see Figure~\ref{fig:split}.
Let $o_p$ be a topological order of the resulting directed acyclic 
graph, which underlies  the split relation on $S_c$ induced by $p$.


The following lemma shows that $\vld(C)$ exists
by construction. 
It builds upon a more restricted version regarding a boundary
curve that had been considered in~\cite{JP19}.

\begin{lemma}
\label{lem:existance}
Given the topological ordering of the split relation $o_p$, 
$\vld(C)$
can be constructed in $O(|C|^2)$ time; thus, $\vld(C)$ exists.
Further, at the same time,
we can construct $\vld(C')$  for any other
Voronoi-like cycle $C'$ 
that is enclosed by $C$, $S_{C'}\subseteq S_C$.
\end{lemma}


\begin{proof}
Given the order $o_p$, we follow the randomized approach of
Chew~\cite{Chew90}, and apply the arc-insertion operation of
Lemma~\ref{lem:insertion}, which is extracted from~\cite{JP18}.
Let the sites in $S_C$ be numbered according to 
$o_p=(s_1,s_2\ldots s_m)$, $m=|S_C|$.
We first show that $C$ can be constructed incrementally, by arc-insertion,
following $o_p$.

Let $C_i$ denote the $p$-cycle constructed by the first
$i$ sites in $o_p$. 
$C_1$ consists of $J(s_1,p)$ and a $\Gamma$-arc, 
that is, $C_1=\partial(D(p,s_i)\cap
D_\Gamma)$. Clearly $C_1$ encloses $C$.
Suppose  that $C_i$ encloses $C$.
Then, given $C_i$, let $C_{i+1}$ by the $p$-cycle obtained
by inserting to $C_i$ the components of $J(s_{i+1},p)\cap D_{C_i}$, which
correspond to arcs in $C$.
For each such component $\alpha$ ($\alpha\in J(s_{i+1},p)\cap
D_{C_i}$), if some portion of $\alpha$ appears in $C$, then compute
$C_i\oplus \alpha$; if $\alpha$ does not appear in $C$, ignore
it.
Let $C_{i+1}$ be the resulting
$p$-cycle after all such components of $J(s_{i+1},p)$ have been inserted to $C_i$, one by one.
Because  any site whose $p$-bisector splits 
$J(p,s_{i+1})$ has already been processed, a distinct component
of $J(s_{i+1},p)\cap D_{C_i}$ must exist for each arc of 
$s_{i+1}$ in $C$.
Thus, $C_{i+1}$ can be derived from $C_i$ and must enclose $C$.


We have shown that $C$ can be constructed incrementally, if we follow
$o_p$, in time  $O(|C|^2)$.
It remains to construct the Voronoi-like graph
$\vld(C_i)$ at each step $i$.
To this end, we use Lemma~\ref{lem:insertion}, starting at 
$\vld(C_1)=\emptyset$.
Given $\vld(C_i)$ and $C_{i+1}$,  we can apply
Lemma~\ref{lem:insertion} to each arc $\alpha \subseteq J(s_{i+1},p)$ in $C_{i+1}\setminus
C_i$.
The correctness proof of \cite{JP18} ensures the feasibility and the correctness of each
arc insertion, thus, it also ensures the existence of $\vld(C_{i+1})$.

The above incremental construction can also
compute $\vld(C')$ 
by computing both $C_i$ and $C_i'$
at each step $i$.
Suppose $C_i=C_i'$, where $C_1'=C_1$.
When considering site $s_{i+1}$, we
insert to  $C_i'$ all components of $J(s_{i+1},p)\cap D_{C_i}$ 
corresponding to arcs of $s_{i+1}$, which appear in either  $C$ or 
$C'$. 
Thus,  $C_{i+1}'$ is derived from  $C_{i+1}$ by inserting any additional
arcs $\alpha'$ of $s_{i+1}$, where  $\alpha'\in C'\setminus C$.
Note that all arcs of
$s_{i+1}$ that appear in $C$ are inserted to $C_{i+1}'$, even if they do not
appear in $C'$.
This is possible because of the
order $o_p$: any site whose $p$-bisector splits 
$J(p,s_{i+1})$ has already been processed, thus, a distinct component
of $J(s_{i+1},p)\cap D_{C_i}$ must exist for each arc of
$s_{i+1}$ in either $C$ or $C'$, which can be identified.

Referring to $\vld(C_{i+1})$, the insertion of an additional arc
$\alpha'$ may only cause an existing region to shrink.
Therefore, we derive two invariants: 1. $R(\beta,C_{i+1}')\subseteq R(\beta,C_{i+1})$ for any arc $\beta\in
C_{i+1}\cap C_{i+1}'$; and 2.  $C_{i+1}'$ is enclosed by
$C_i$. 
The invariants are maintained in subsequent steps.
The fact that step $i+1$ starts with $C_{i+1}'$, which is enclosed by
$C_{i+1}$, does not make a difference to the above arguments.
Thus, the invariants hold for
$C_n$ and $C_n'$, therefore, $C_n'=C'$.
\end{proof}

The following lemma can also be extracted from \cite{JP18,JP20arxiv}. 
It  can be used to
establish the uniqueness of $\vld(C)$.
Similarly to Lemma~\ref{lem:insertion}, its original statement 
does not refer to a $p$-cycle, 
however, nothing in its proof prevents its adaptation 
to a $p$-cycle, 
see~\cite[Lemma~29]{JP20arxiv}.

\begin{lemma} \cite{JP20arxiv}
 \label{lem:missingarc}
Let $C$ be a $p$-cycle and let $\alpha,\beta$ be two bisector arcs in
$C$, where $s_\alpha\neq
s_\beta$.
Suppose that a component $e$ of $J(s_\alpha,s_\beta)$ intersects  $R(\alpha,C)$.
Then $J(p,s_\beta)$ must intersect $D_c$ with a component
$\beta'\subseteq J(p,s_\beta)\cap D_c$ such that $e$ is a portion of
$\partial R(\beta',C\oplus\beta')$.
\end{lemma}

By Lemma~\ref{lem:missingarc},  if $J(s_\alpha,s_\beta)$ intersects
$R(\alpha,C)$, then a face of $s_\beta$ must be \emph{missing} from
$\vld(C)$ (compared to $\vld(\hat C)$) implying that an arc of $J(p,s_\beta)$ is missing
from $C$. Then $\vld(C)$ must be unique.

We now use the randomized incremental construction of \cite{JP18} to
construct $\vld(C)$, which in turn follows Chew~\cite{Chew90},
to establish the last claim of Theorem~\ref{thrm:vld}.
Let $o=(\alpha_1,\ldots \alpha_n)$ be a random permutation of the bisector arcs
of $C$, where each arc represents a different occurrence of a site in $S_C$.
The incremental algorithm works in two phases.
In phase~1, delete arcs from $C$
in the reverse order $o^{-1}$, 
while registering their neighbors at the
time of deletion.
In phase~2,  insert the arcs one by one, following
$o$, using their neighbors information from phase~1. 

Let $C_i$ denote the $p$-cycle constructed by considering the
first $i$ arcs in $o$ in this order.
$C_1$ is the $p$-cycle consisting of $J(s_{\alpha_1},p)$ and the relevant
$\Gamma$-arc. 
Given $C_i$, let $\alpha_{i+1}'$ denote the bisector
component of $J(p,s_{\alpha_{i+1}})\cap D_{C_i}$ that contains
$\alpha_{i+1}$ (if any), see Figure~\ref{fig:cases} where $\alpha$
stands for $\alpha_{i+1}'$. 
If $\alpha_{i+1}$ lies outside $C_i$, then
$\alpha_{i+1}'=\emptyset$
(this is possible if $C_i$ is not a Voronoi cycle).
Let cycle
$C_{i+1}=C_i\oplus\alpha_{i+1}'$ (if $\alpha_{i+1}'=\emptyset$,
then  $C_{i+1}=C_i$).
Given $\alpha_{i+1}'$, and  $\vld(C_i)$, the graph 
$\vld(C_{i+1})$ is obtained
by applying Lemma~\ref{lem:insertion}.


Let us point out a critical case, which differentiates from \cite{Chew90}: both endpoints of
$\alpha_{i+1}'$ lie on the same arc $\omega$ of $C_i$, see
Figure~\ref{fig:cases}(c) where $\alpha$ stands for $\alpha_{i+1}'$.
That is, the insertion of $\alpha_{i+1} $ splits the arc $\omega$ in two arcs, $\omega_1$ and $\omega_2$.
(Note $ s_{\alpha_{i+1}} <_p s_\omega$ 
but $\omega$ was inserted to $C_i$ before  $\alpha_{i+1}$).
Because of this split, $C_i$, and thus $\vld(C_i)$, is order-dependent: if $\alpha_{i+1}$ were
considered before $\omega$, in some alternative ordering, then $\omega_1$ or $\omega_2$ would not
exist in the resulting cycle, and similarly for their faces in
$\vld(C_{i+1})$.
The time to split $R(\omega,C_i)$ is proportional to the minimum complexity of
$R(\omega_1,C_{i+1})$ and $R(\omega_2,C_{i+1})$, which is added to the
time complexity
of step $i$.
Another side effect of the split relation is that $\alpha_{i+1}$ may fall outside
$C_i$, if $C$ is not a Voronoi-cycle,  in which case,  
$C_{i+1}=C_i$.
Then  $C_n\neq C$, in particular, $C_n$
is enclosed by $C$. 

Because the computed cycles are order-dependent, standard backwards analysis
cannot be directly applied to 
step $i$.
In \cite{JP20arxiv} an alternative technique was proposed, 
which can 
be applied to the above construction.
The main difference from
\cite{JP20arxiv} is 
case  $C_{i+1}=C_i$,
however, such a case has no effect to  time complexity,
thus, the analysis of 
\cite{JP20arxiv} can be applied.

\begin{proposition}
By the variant of backwards analysis in \cite{JP20arxiv}, the time complexity of step $i$ 
 is expected $O(1)$.
\end{proposition}


\subsection{The relation among the Voronoi-like graphs
  $\vld(C)$, $\vld(C')$, and $\vld(\hat C)$}

 In the following proposition, the first claim follows from
 Theorem~\ref{thrm:face} and the second follows from the proof of
 Lemma~\ref{lem:existance}.

\begin{proposition}
\label{cor:subset}
Let $C'$ be a Voronoi-like cycle between $C$ and $\hat C$ such that $S_{\hat
  C}\subseteq S_{C'}\subseteq S_{C}$. 
\begin{enumerate}
\item  $R(\alpha,C') \supseteq R(\alpha,\hat C)$, for any arc $\alpha\in C'\cap
\hat C$.
\item  $R(\alpha,C') \subseteq R(\alpha,C)$, for any arc $\alpha\in C\cap
C'$.
\end{enumerate}
\end{proposition}

Proposition~\ref{cor:subset} indicates that the faces  of
$\vld(C')$ \emph{shrink} as we move from the outer cycle $C$ to an inner
one, 
until we reach the Voronoi faces of $\vld(\hat C)$, which are
contained in all others.
It also indicates that $\vld(C)$, $\vld(C')$ and $\vld(\hat C)$ 
share common subgraphs, and that the adjacencies of the Voronoi diagram
$\vld(\hat C)$ are  preserved. More formally, 

\begin{definition}
Let $\vld(C',C\cap C')$ be
the following subgraph of $\vld(C')$: vertex  $v\in \vld(C')$ is included in
$\vld(C',C\cap C')$, if all three faces incident to $v$ belong to 
arcs in  $C\cap C'$;
edge $e \in \vld(C')$ is included to $\vld(C',C\cap C')$ if
both faces incident to $e$ belong to arcs in  $C\cap C'$.
\end{definition}

\begin{proposition}
\label{cor:subset2}
For any Voronoi-like cycle $C'$,
enclosed by $C$, where 
$S_{C'}\subseteq S_{C}$,
it holds: 
$\vld(C',C\cap C') \subseteq  \vld(C)$. 
\end{proposition}

Depending on the problem at hand, computing  $\vld(C')$ (instead of the more expensive task of computing $\vld(C)$)
may be sufficient. For an example see
Section~\ref{sec:order-k}.

Computing $\vld(C)$ in linear expected time, instead of $\vld(C')$,
is possible if the faces of $\vld(C)$ are Voronoi
regions.
This can be achieved by deleting the superflous arcs in $C'\setminus
C$, created during the arc-splits,
which are called \emph{auxiliary arcs}. 
A concrete example is given in Section~\ref{sec:CDT}.
During any step of the construction, if $R(\alpha',C_i)$ is
 a Voronoi region,  but $\alpha'\cap C=\emptyset$, 
 we can call
the site-deletion procedure of \cite{JP18}
to eliminate $\alpha'$  and $R(\alpha',C_i)$ 
from $\vld(C_i)$. In particular,

\begin{proposition}
\label{cor:del}
Given $\vld(C_i)$, $1\leq i\leq n$, we can delete $R(\alpha,C_i)$, if
$R(\alpha,C_i) \subseteq \VR(s_{\alpha}, S_{\alpha})$, where
$S_{\alpha}\subseteq S_C$ is the set of sites that define 
$\partial R(\alpha,C_i)$,
in expected time linear on $|S_{\alpha}|$.
\end{proposition}

There are two ways to use Proposition~\ref{cor:del}, if applicable:
\begin{enumerate}
\item Use it when necessary to 
maintain the invariant that $C_i$ encloses $C$ (by deleting \cite{JP18} any
auxiliary arc in $C_{i-1}$ that blocks the insertion of  $\alpha_i$, thus, eliminating the case $C_i =C_{i-1}$).
\item Eliminate any auxiliary arc at the time of its creation. 
If the insertion of $\alpha_i$ splits an arc $\omega\in C_{i-1}$
into $\omega_1$ and $\omega_2$,  but $\omega_2\not\in C$, then
eliminate $R(\omega_2,C_i)$ by calling \cite{JP18}.
\end{enumerate}

The advantage of the latter is that Voronoi-like cycles become
order-independent, 
therefore,
backwards analysis becomes possible to establish the algorithm's time complexity.
We give the backwards analysis argument on the concrete
case of Section~\ref{sec:CDT}; the same type of argument, only
more technical, can be derived for this abstract formulation as well.



\section{Extending to Voronoi-like cycles of $k$ sites}
\label{sec:order-k}

Theorem~\ref{thrm:vld} can extend to a \emph{Voronoi-like $k$-cycle},
for brevity, a \emph{$k$-cycle}, which involves 
a set $P$ of $k$ sites whose labels appear
in the interior of the cycle.
A $k$-cycle $C_k$ lies in the arrangement $\arr(\J_P\cup
\Gamma)$  and its vertices are degree-2 locally Voronoi, where $\J_P$
denotes the set of bisectors related to the sites in $P$. 
It implies a Voronoi-like graph $\vld(C_k)$, which involves the
set of sites $S_C\subseteq S\setminus P$, which (together
with the sites in $P$) define  the bisector arcs of $C_k$.
$\vld(C_k)$ is defined  analogously to
Def.~\ref{def:vldC}, given $C_k$ and the set of sites $S_C$.

We distinguish two different types of $k$-cycles on $\arr(\J_P\cup
\Gamma)$:
1. a \emph{$k$-site Voronoi-like
  cycle} whose vertices are  all 
of the nearest type, e.g., the boundary  of the union of $k$
neighboring Voronoi regions; and
2. an \emph{order-$k$ Voronoi-like
  cycle} whose vertices are both of the nearest and the farthest type,
e.g.,  the boundary of an order-$k$ Voronoi face.
%
%
In either case we partition a $k$-cycle $C_k$ into maximal 
\emph{compound arcs}, each induced by one site in $S_C$.
Vertices in the interior of a compound arc are switches between
sites in $P$, and the endpoints of compound arcs are switches between
sites in $S_c$.
For an order-$k$ cycle, the former vertices 
are of the farthest type, whereas the latter
(endpoints of compound arcs) are of the nearest type.
Given a compound arc $\alpha$, let $J(\alpha)$ denote the bisector
curve that consists of the arc $\alpha$ extending the 
bisector arcs incident to its endpoints to $\Gamma$, see Figure~\ref{fig:J(a)}.
Let $P_\alpha\subseteq P$ be the subset of sites 
that (together with one site in $S_C$) define $\alpha$.

\begin{figure}
 		\centering
		\includegraphics[scale=0.8]{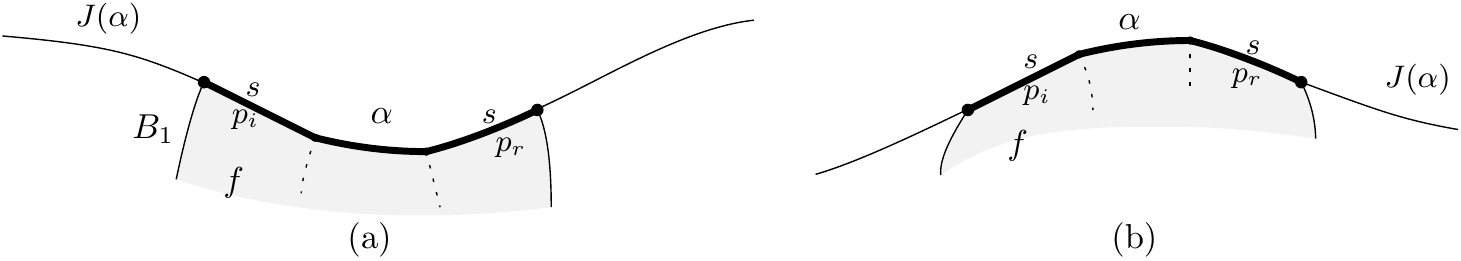}
		\caption{Bisector $J(\alpha)$ for a compound arc $\alpha$ in a $k$-site cycle (a)
                  and  an order-$k$ cycle (b).}
		\label{fig:J(a)}
                  \end{figure}

\begin{lemma}
\label{lem:vld-k}
Assuming that it exists, $\vld(C_k)$ is a forest, and if $C_k$ is bounded,
then $\vld(C_k)$ is a tree.
Each face of $\vld(C_k)$ is incident to exactly one compound arc 
$\alpha$ of
$C_k$, which is denoted as $R(\alpha,C_k)$.
\end{lemma}


\begin{proof}
  $\vld(C_k)$ may not contain cycles because $\VR(s,S_C), s\in S_C$,
  cannot be enclosed by $C_k$, as 
  in the proof of 
  Lemma~\ref{lem:vld}. For the same reason, any face of
  $\vld(C_k)$ must be incident to a bisector arc.
  %
  Thus,  $\vld(C_k)$ is a forest whose leaves are incident to
  the endpoints of 
  compound arcs.
  It remains to show that no face of  $\vld(C_k)$  can be incident to
  a pair of compound arcs  of the same site $s\in S_c$. 

  Suppose, for the sake of contradiction,  that a face $f$ is incident
  to two compound arcs 
  $\alpha,\alpha'\in C_k$ of the same site $s\in S_C$ ($s=s_f$).
  We first consider an order-$k$ cycle, see 
  Figure~\ref{fig:lemma:vld-k}. 
  Arcs $\alpha$ and $\alpha'$ consist of bisector pieces in
  $J(s,p_i)$, $p_i\in P$.
  Any two of these $s$-related bisectors $J(s,p_i)$, $J(s,p_j)$ 
  must intersect at least once, as otherwise 
  $\VR(s,\{s,p_i,p_j\})$ would be disconnected, violating axiom A1.
  Furthermore, any two $J(s,p_i)$ and  $J(s,p_r)$ 
  contributing to the same compound arc must intersect
  exactly once, because if they intersected twice, they would intersect
  under an illegal pattern of
  Figure~\ref{fig:bisector-2-intersection}(d), see Figure~\ref{fig:lemma:vld-k}(c).
 
 \begin{figure}
 		\centering
		\includegraphics[scale=0.85]{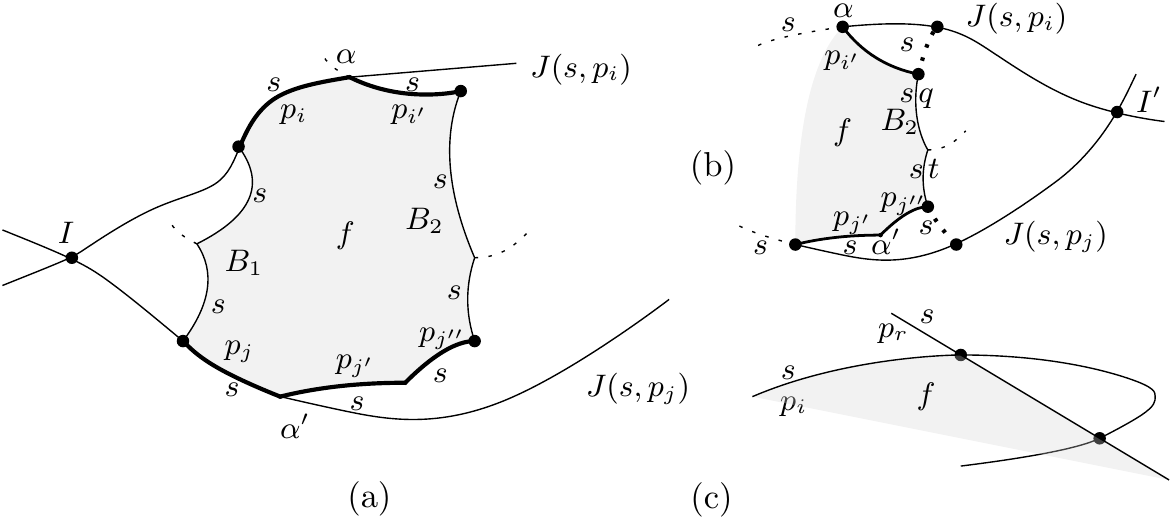}
		\caption{Proof of Lemma~\ref{lem:vld-k} for an
                  order-$k$ Voronoi-like cycle; face $f$ is
                  shown shaded.}
		\label{fig:lemma:vld-k}
              \end{figure}
              
  Consider the two branches of  $\partial
          f\setminus\{\alpha,\alpha'\}$, see
          Figure~\ref{fig:lemma:vld-k}. Choose one such brunch, say
          $B_1$, 
          and let $\alpha_i\subset J(s,p_i)$ and  $\alpha_j\subset
          J(s,p_j)$ be the bisector arcs of $\alpha$ and $\alpha'$
          respectively incident to the endpoints of $B_1$.
          If  $J(s,p_i)$ and  $J(s,p_j)$ intersect at a point $I$ at
          opposite side of $B_1$ as $\alpha_i$ and $\alpha_j$, then we
          have a cycle formed by $B_1$ and the pieces of $J(s,p_i)$
          and  $J(s,p_j)$ incident to $I$  that has the label $s$
          outside.
          But such a cycle cannot exist, by Observation~\ref{obs:inversecycle}.
          Thus, $I$ cannot exist  and  $J(s,p_i)$, $J(s,p_j)$ must
          intersect at a point $I'$ on the other side of $B_1$.

          Bisector $J(s,p_i)$ (resp. $J(s,p_j)$)  cannot enter face $f$ because
          otherwise $J(s,p_i)$ (resp. $J(s,p_j)$) would intersect
          twice with another $s$-related bisector contributing to
          arc $\alpha$ (resp. $\alpha')$, which is not possible as
          claimed above. Thus, 
          $I'$ 
          cannot lie within $f$.

    
          Consider the other brunch $B_2$ of  $\partial
          f\setminus\{\alpha,\alpha'\}$ and expand the arcs 
          incident to its endpoints until one hits  $J(s,p_i)$ and the
          other hits $J(s,p_j)$, see Figure~\ref{fig:lemma:vld-k}(b).  The bisectors constituting
          $B_2$ are $s$-related, thus, they must intersect $J(s,p_i)$
          and $J(s,p_j)$, as otherwise the illegal pattern of
          Figure~\ref{fig:bisector-0-intersection}(b) would appear.
          Suppose now that $J(s,p_i)$ and  $J(s,p_j)$ intersect at a point $I'$ at
          the opposite side of $B_2$ as $f$.
          Then an illegal cycle with the label $s$ outside is
          constructed by the expanded brunch $B_2$ and the pieces of $J(s,p_i)$  
          and  $J(s,p_j)$ incident to $I'$, concluding that $I'$ is
          not possible either, by Observation~\ref{obs:inversecycle}.
          We derive a contradiction as $J(s,p_i)$
          and $J(s,p_j)$ must intersect at least once.
          Thus, each face of $\vld(C_k)$ must be incident to exactly one
          order-$k$ arc of $C_k$.

          Suppose now that $C_k$ is a $k$-site Voronoi-like
          cycle and face $f$ is incident to compound arcs $\alpha$ and
          $\alpha'$. Consider the curves $J(\alpha)$  and
          $J(\alpha')$, which  can not 
          intersect $B_1$ nor $B_2$ because
          otherwise an illegal cycle, having the label
          $s$ outside, would be created  contradicting
          Observation~\ref{obs:inversecycle}.
          (In Figure~\ref{fig:J(a)}(a) an illegal cycle would be
          created if  $J(\alpha)$ turned to intersect $B_1$).
          Furthermore, $J(\alpha)$ and $J(\alpha')$  must intersect otherwise
          $\VR(s,P_\alpha\cup P_{\alpha'}\cup\{s\})$ would be
          disconnected. 
          But then an illegal
          cycle,  with the label
          $s$ outside, would be created between the intersecting pieces
          of  $J(\alpha)$ and $J(\alpha')$, and  $B_1$ or $B_2$,
          contradicting Observation~\ref{obs:inversecycle}.
\end{proof}

Given Lemma~\ref{lem:vld-k}, the remaining claims of
Theorem~\ref{thrm:vld} can be derived as in  Section~\ref{sec:cycle}.
Let $J(s_\alpha,P)$ denote the bisector curve associated with a
compound arc $\alpha$, $s_\alpha\in S_C$.
For a $k$-site cycle,
$J(s_\alpha,P)=\partial\VR(s_\alpha,P\cup\{s_\alpha\})$.
%
For an order-$k$ cycle,  $J(s_\alpha,P)=\partial
\freg(\alpha,P\cup\{s_\alpha\})$, where
$\freg(\alpha,P\cup\{s_\alpha\})$ denotes the face of the
farthest Voronoi region of $s_\alpha$, which is incident to arc
$\alpha$. 
In both cases $J(\alpha)=J(s_\alpha,P_\alpha)$.

The curve  $J(s_\alpha,P)$ is
expensive to compute, however, we never need to entirely compute it.
Instead of $J(s_\alpha,P)$, we use $J(s_\alpha,\tilde P_\alpha)$, where
$P_\alpha\subseteq \tilde P_\alpha\subseteq P$, and  $|P_\alpha| \leq
|\tilde P_\alpha| \leq |P_\alpha|+2$.
$J(s_\alpha,\tilde P_\alpha)$ is readily available from
$J(\alpha)$  and the two neighbors of $\alpha$ at its insertion time in the 
current Voronoi-like cycle.
Using $J(s_\alpha,\tilde P_\alpha)$  in the place of the $p$-bisectors
of Section~\ref{sec:vld} the same essentially incremental algorithm
can be applied  on the compound arcs of $c_k$.
Some properties of  $J(s_\alpha,\tilde P_\alpha)$ in the
case of an order-$k$ cycle are given in \cite{JP20arxiv}.

\subsection{Computing a Voronoi-like graph in an order-$k$ Voronoi face}
We now review an example by Junginger and Papadopoulou~\cite{JP19} when $C_k$ is the boundary
of  a face $f$ of an order-$k$ Voronoi region.
It is  known that $\vld(C_k)$  can be computed in
linear-expected time~\cite{JP20arxiv},
but an even simpler technique can be derived by computing the
Voronoi-like graph of an appropriately defined Voronoi-like cycle $C$~\cite{JP19}.
In fact, 
the Voronoi-like graph of any Voronoi-like cycle  $C'$, between $C$ and
$\hat C$, turns out fully sufficient. 

Let $f$ be a 
face of an order-$k$
Voronoi region of a set $H$ of $k$ sites.
Let $S_f$ denote the set of sites that, together with the sites in
$H$, define the boundary $\partial f$.
The graph $\vld(\partial f)$ gives  the order-$(k{+}1)$
Voronoi subdivision within $f$, which is the Voronoi diagram
$\V(S_f)$, truncated within $f$, i.e., $\V(S_f)\cap f$.



\begin{description} 
\item Computing the Voronoi diagram $\V(S_f)\cap f= \vld(\partial f)$ \cite{JP19}.
\begin{itemize}
\item Given $\partial f$, and any $h\in H$,  compute an $h$-cycle $C$ as
implied by the order of sites along the boundary of $f$. Note that $C$
encloses the Voronoi region $\VR(h,S_f)$, which in turn encloses
$f$. $\VR(h,S_f)$ is not known, however, $C$ can be derived directly 
from $\partial f$.

\item Run the randomized incremental technique of \cite{JP18}
 on $C$ in linear
  expected time  (see Section~\ref{sec:cycle}). It
will compute $\vld(C')$
for some $h$-cycle between $C$ and $\hat C$.
\item Truncate $\vld(C')\cap f$. No matter which $h$-cycle is computed, $\vld(C')\cap f=
  \V(S_f)\cap f$.

\end{itemize}
\end{description}

The  claim follows by  
the fact that $R(\alpha',\hat C)
\cap f=\emptyset$, for any $\alpha'\in \hat C\setminus C'$, and
$C\setminus C'\subseteq \hat C\setminus C$. Thus,
$\vld(\hat C)\cap f= \vld(\hat C, \hat C\cap C')\cap f = \vld(\hat C, \hat C\cap C)\cap f$.

\section{Updating a constraint Delaunay
       triangulation}\label{sec:CDT}
      We give an example of a Voronoi-like cycle $C$, which does  not
     correspond to a Voronoi
     region, but we need to compute the adjacencies of the Voronoi-like
     graph $\vld(C)$.
     The problem appears in the incremental construction of a \emph{constraint Delaunay
       triangulation} (CDT), a well-known variant of the Delaunay
     triangulation, in which a given set of segments is constrained
     to appear in the triangulation of a point set $Q$, which includes
     the endpoints of the segments, see \cite{SB15} and references therein.
     
     Every edge of the CDT is either an input segment or is
     \emph{locally Delaunay} (see Section~1).
     The incremental construction to compute a CDT,
     first constructs an ordinary
       Delaunay triangulation of the points in $Q$, and then inserts segment
      constraints, one by one, updating the triangulation after
      each insertion.
      Shewchuk and  Brown~\cite{SB15} gave an expected
      linear-time algorithm to perform each update. 
      Although the algorithm is summarized in a pseudocode, which could then
      be directly implemented,
      the algorithmic description is quite
      technical having to make sense of self-intersecting polygons,
      their triangulations, and other exceptions.
      We show that the problem corresponds exactly to computing (in
      dual sense) the
      Voronoi-like graph of a Voronoi-like cycle. 
      Thus, a very simple randomized incremental construction, with
      occasional calls to Chew's algorithm~\cite{Chew90} to delete a
      Voronoi region of points, can be derived.
      Quoting from \cite{SB15}: incremental segment insertion is
      likely to remain the most used CDT construction algorithm, so it
      is important to provide an understanding of its performance and
      how to make it run fast.
      We do exactly the latter in this section.


        \begin{figure}
	\centering 
	\includegraphics[scale=0.9]{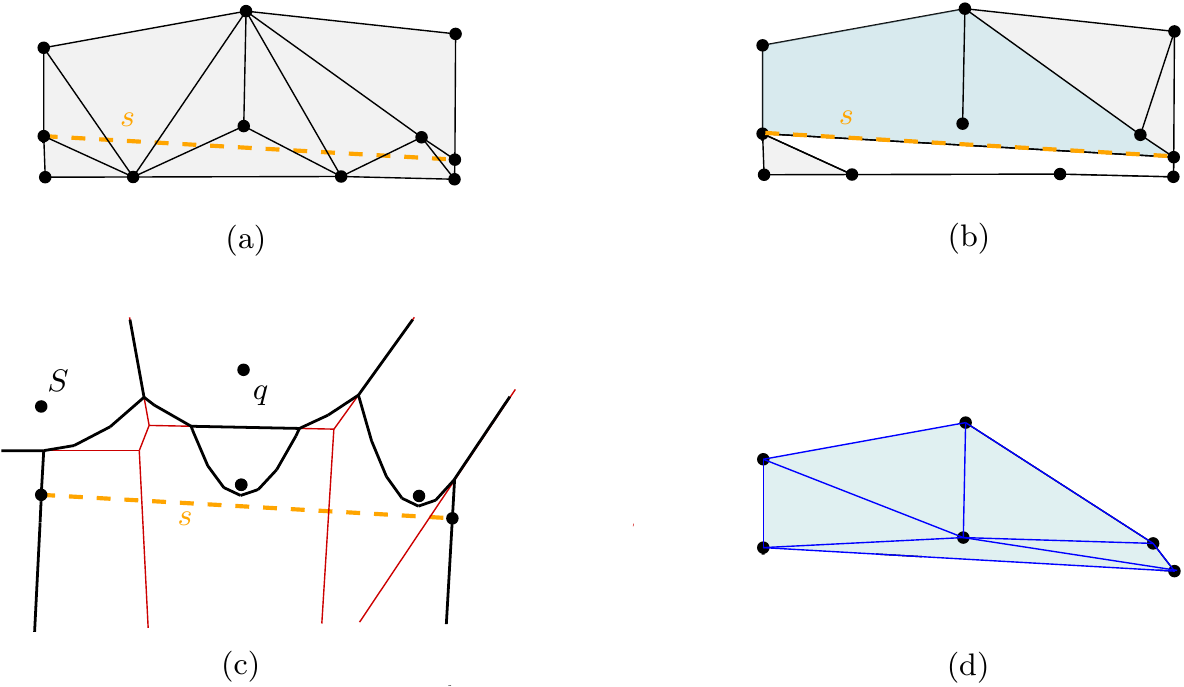}
	\caption{Point set from \cite[Fig.~3]{SB15}. (a) The
          given CDT and the segment $s$ superimposed; (b) the cavity
          $P$ in blue; (c) $\vld(C)$ in red,  $C=\partial \VR(s,
          S\cup \{s\})$ in black; (d) the re-triangulated  $P$.
        }
	\label{fig:cdt}
      \end{figure}

       When a new constraint segment $s$ is inserted in a CDT, the
       triangles, which get 
       destroyed by that segment, are  identified and
       deleted~\cite{SB15}. This creates two \emph{cavities} that need to be
       re-triangulated using \emph{constrained Delaunay triangles},
       see Figure~\ref{fig:cdt}(a),(b), borrowed from \cite{SB15},
       where one cavity is shown shaded (in light blue) and the other unshaded.
       The boundary of each cavity need not 
       be  a simple polygon. 
       However, each cavity implies a Voronoi-like cycle,
       whose Voronoi-like graph re-triangulates the cavity, see Figure~\ref{fig:cdt}(c),(d).

        \begin{figure}[h]
	\centering 
	\includegraphics[scale=0.9]{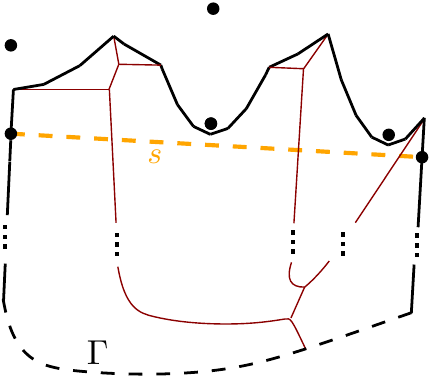}
        \caption{The Voronoi-like cycle $C$ and $\vld(C)$ (in red) for the
          example of Fig.~\ref{fig:cdt}.}
	\label{fig:cdt-vld}
      \end{figure}

      Let  $P=(p_1,p_2,\ldots, p_n)$ denote one of the cavities, where $p_1 \ldots p_n$
      is the sequence of cavity vertices in counterclockwise order, and  $p_1,p_n$ are the endpoints of $s$.
      Let $S$ denote the corresponding set of points ($|S|\leq n$) and
      let $\J_s$ denote the underlying bisector system
      involving the segment $s$ and points in $S$.
      Let  $C$ be the $s$-cycle in $\arr(\J_S\cup\Gamma)$ that has one $s$-bisector
      arc for each  vertex in $P$, in the same order as $P$, see
      Figure~\ref{fig:cdt-vld}.
      Note that one  point in $S$ may contibute more than one arc in
      $C$.

    \begin{lemma}
        \label{lem:C}
      The  $s$-cycle $C$ exists and can be derived from $P$ in linear time.
    \end{lemma}

       \begin{figure}
	\centering 
	\includegraphics[scale=0.9]{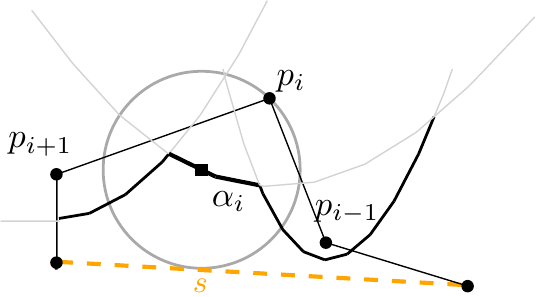}
        \caption{Proof of Lemma~\ref{lem:C}.}
	\label{fig:lem:C}
      \end{figure}
      
    \begin{proof}
      Let $p_i\in P, 1<i<n$. The diagonal of the original CDT, which bounded the triangle  incident to
       $p_ip_{i-1}$ (resp. $p_ip_{i+1}$) was a
       locally Delaunay edge intersected by $s$.
      Thus, there is a
      circle through $p_i$ that is tangent to $s$ that contains 
      neither $p_{i-1}$ nor $p_{i+1}$, see Figure~\ref{fig:lem:C}. Hense, an arc of $J(p_i,s)$ must
      exist, which contains the center of this circle, and extends
      from an intersection point of $J(p_i,s)\cap
      J(p_{i-1},s)$ to an intersection point of $J(p_i,s)\cap
      J(p_{i+1},s)$. The portion
      of $J(p_i,s)$ between these two intersections corresponds to the
      arc of $p_i$ on $C$, denoted $\alpha_i$.
      Note that the $s$-bisectors are parabolas that share the same directrix
      (the line through $s$), thus, they may intersect twice.
      It is also possible that  $p_{i-1}=p_{i+1}$. In each case,
      we can determine which intersection  is relavant to arc
      $\alpha_i$, given the counterclockwise order of $P$.
      Such questions can be reduced to \emph{in-circle tests} involving the
      segment $s$ and three points.
     \end{proof}

     Let $\cdt(P)$  denote the  constraint Delaunay triangulation of
     $P$. Its edges are either locally Delaunay or they are cavity edges on
     the boundary of $P$.
  
     \begin{lemma}
       The $\cdt(P)$ is  dual to 
       $\vld(C)$, where $C$ is the $s$-cycle derived from $P$.
     \end{lemma}

     \begin{proof}
      The claim derives from the definitions, Lemma~\ref{lem:C}, which
      shows the existence of $C$,
      and the properties of 
      Theorem~\ref{thrm:vld}. 
      The dual of $\vld(C)$ has one node for each $s$-bisector arc of $C$,
      thus, one node per vertex in $P$.
      An edge of $\vld(C)$ incident to two locally Voronoi
      vertices $v,u$  involves four different sites in
      $S$; thus, its dual edge is locally Delaunay.
      The dual of an edge incident to a leaf of $C$, 
      is an edge of the
      boundary of $P$. 
     \end{proof}

Next, we 
compute $\vld(C)$ in expected linear time.
Because $C$ is not the complete boundary of a Voronoi-region,
if we apply the construction of Theorem~\ref{thrm:vld},  
the computed cycle $C_n$ may be 
enclosed by $C$.
This is because of occasional split operations, given the random order
of arc-insertion,
which may create \emph{auxiliary arcs} that have no correspondence to vertices of
$P$.
However, we can use Proposition~\ref{cor:del} to delete such auxiliary
arcs and their faces. 
The sites in $S$ are points, thus, any Voronoi-like cycle
in their bisector arrangement coincides with a Voronoi region. 
By calling Chew's algorithm~\cite{Chew90} we can delete any face of any auxiliary
arc in expected time linear in the complexity of the face.

It is easy to dualize the technique to directly compute constraint
Delaunay triangles.
In fact, the cycle $C$ can remain  conceptual with no need to explicitly compute it. 
The dual  nodes  are  graph
theoretic, each one corresponding to an  $s$-bisector  arc,
which in turn corresponds to a cavity vertex.
This explains the
polygon self-crossings of \cite{SB15} if we 
draw these graph-theoretic  nodes on the cavity vertices
during the intermediate steps of the construction.


The algorithm to compute $\vld(C)$ (or its dual $\cdt(C)=\cdt(P)$) is very simple.
Let  $o=(v_1,\ldots v_n)$ be a random permutation of the vertices in
$P$, except the endpoints of $s$; let $v_1=p_1$ and
$v_2=p_n$.
Let  $P_i$ denote the sub-sequence of $P$ consisting of the first
$i$ vertices in $o$.
Let  $C_i$ denote the corresponding $s$-cycle, which has one $s$-bisector arc for each vertex in $P_i$ in the
order of $P_i$ (see Lemma~\ref{lem:C}).
In an initial  phase~1, starting at $P_n=P$, delete
vertices in reverse order $o^{-1}$, recording the
neighbors of each vertex $v_i$ in $P_i$ at the time of its deletion.
%
In phase~2, 
consider the vertices in $o$ in increasing order, starting with
$\vld(C_3)$, and using the arc-insertion operation 
(Lemma~\ref{lem:insertion})
to build $C_i$ and 
$\vld(C_i)$
incrementally, $3\leq i\leq n$.
Instead of $\vld(C_i)$, we can equivalently be constructing the dual $\cdt(C_i)$.

In more detail, let $C_3$ be the $s$-cycle obtained by the two perpendicular lines
through the endpoints of $s$, which are truncated on one side by $\Gamma$, and on
the other by  $J(v_3,s)$.  $C_3$ consists of four arcs on: $J(s,p_1), J(v_3,s),
J(s,p_n)$ and $\Gamma$, respectively. $\vld(C_3)$ has one Voronoi vertex
for $(p_1,v_3, p_n)$, see Figure~\ref{fig:cdt-steps}(a). 

Given $\vld(C_{i-1})$, we insert $v_{i}$  between its two neighboring vertices
$w,u$, which have been recorded in phase~1. Suppose $w,v_i,u$ appear
in counterclockwise order in $P$, see Figure~\ref{fig:charge}(a), where $v_i=v$. 
Let $\alpha_{i}$  denote the arc of $v_{i}$  in $C_{i}$, in particular, $\alpha_{i}$ is the component of $J(v_{i},s)\cap D_{C_{i-1}}$
whose endpoints lie between the arcs of $u$ and $w$ in $C_{i-1}$,
call them $\beta$ and $\omega$ respectively, see
Figure~\ref{fig:charge}(a), where $\alpha_i=\alpha$.
Among the three cases of the arc
insertion operation, we only consider  the split case (depicted in
Figure~\ref{fig:cases}(c) and~\ref{fig:charge}(a)),
where
$J(v_{i},s)$ \emph{splits} (intersects twice) the arc $\omega\subseteq
J(w,s)$ in $C_{i-1}$; the other cases are straightforward.
In this case, when inserting $\alpha_{i}$ to  $\vld(C_{i-1})$, the region  $R(\omega,C_{i-1})$ is split in two faces, where 
one, say $R(\omega_2,C_i)$, does not correspond
to $w$ (since it is out of order with respect to  $w,v_{i}, u$).
That is,  we compute  $\vld(C_{i}')$, where
$C_{i}'=C_{i-1}\oplus\alpha_{i}$ and includes the auxiliary arc  $\omega_2$.
%
To obtain  $\vld(C_{i})$ we can call Chew's algorithm to delete $R(\omega_2,C_i')$,
thus, restore $C_{i}$ to its original definition.
The increase to the time complexity of step $i$ 
is expected $O(|R(\omega_2,C_i)|$.
This is 
not covered by the argument of \cite{JP20arxiv}, which proves
the expected constant time complexity of step $i$.
However, 
by deleting auxiliary arcs,  $\vld(C_i)$ becomes
order-independent, therefore, we can  prove the
time complexity of step $i$ in simpler
terms by invoking backwards analysis.


\begin{lemma}
The time complexity of step $i$, which computes $\vld(C_i)$
enhanced by calling Chew's
algorithm to delete any generated auxiliary arc, is expected
$O(1)$. 
\end{lemma}

\begin{proof}
Since $C_{i-1}$ contains no auxiliary arcs, step $i$ can be performed
in time proportional to $|R(\alpha_{i},C_{i}')|$ +
$|R(\omega_2,C_{i}')|$, where $C_{i}'=C_{i-1}\oplus\alpha_{i}$, and
$\omega_2$ is the auxiliary arc when 
inserting $\alpha_{i}$ to  $\vld(C_{i-1})$. 
The first term $|R(\alpha_{i},C_{i}')|\leq
|R(\alpha_{i},C_{i})|$.
The second term can be expressed as 
$|R(\omega_2,C_{i}\oplus \omega_2|)$, i.e, the face complexity of $\omega_2$, if we
insert the arc $\omega_2$ to $\vld(C_{i})$. We charge 1 unit, on
behalf of  $v_{i}$, to any vertex of $\vld(C_{i})$
that would get deleted if we inserted the arc $\omega_2$.

Let $V_{i}=\{v_3\ldots v_{i}\}$. 
Any vertex in $V_{i}$ is
equally likely to be the last one considered at step $i$.
Thus, we can add up the time complexity of step $i$ when considering  
each vertex in $V_{i}$ as last, and take the average.
The total is  $O(|\vld(C_{i})|)$  for the first term, plus the total
number of charges for the second.
By the following lemma the total number of charges is also
$O(|\vld(C_{i})|)$.
Therefore, the average time complexity is $O(1)$.
\end{proof}

     \begin{figure}
	\centering 
	\includegraphics[scale=0.9]{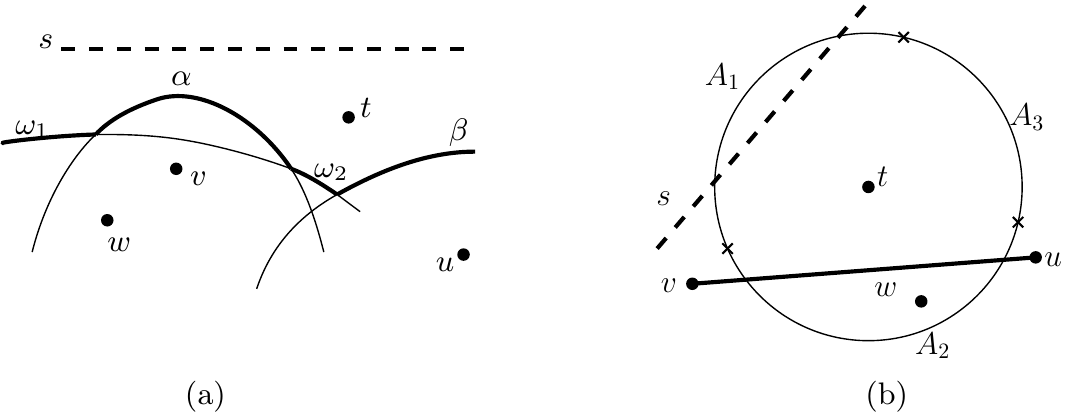}
        \caption{Proof of Lemma~\ref{lem:charge}. }
	\label{fig:charge}
      \end{figure}

\begin{lemma}
\label{lem:charge}
At step $i$,  any vertex  of  $\vld(C_{i})$ can be charged
at most twice. 
\end{lemma}

\begin{proof}
Consider a vertex $t$ of $\vld(C_{i})$ and its Delaunay circle $C_t$
passing through three vertices of $P_{i}$, indicated by crosses in
Figure~\ref{fig:charge}(b). The three vertices
partition $C_t$ in three arcs: $A_1,A_2, A_3$. The segment $s$ must
cross through (intersect twice) one of these arcs, say $A_1$, since
$s$ must be visible to $t$ and the three defining sites of  $C_t$.

Suppose $t$ is charged 
one unit by  $v\in V_{i}$.
Suppose $w,v,u$ appear consecutively counterclockwise around $P_{i}$. 
Let $\omega,\beta$ be the arcs corresponding to  $w$ and $u$,  respectively, in
$C_{i-1}$, see Figure~\ref{fig:charge}(a).
Since $t$ is charged 
one unit by  $v$, it follows that $\omega\in C_{i-1}$ gets split by the insertion
of $v$ creating an auxiliary arc $\omega_2$, and  $t$ lies in $R(\omega_2,C_{i}\oplus \omega_2)$.
That is, $w$ is enclosed by $C_t$ but  $v$ and $u$ are not.
Thus, diagonal $vu$ must intersect $C_t$, and since it cannot obstruct
the visibility between $s$ and the defining points of $C_t$, it must
cross through another arc of $C_t$, say $A_2$; diagonal $uv$ leaves
$w$ and $t$ on opposite sides.
But $s$ must be visible to diagonal $uv$,
thus,  no other diagonal
of $P_i$ can also cross through $A_2$,  obstructing the visibility of
$uv$  and $s$.
Thus, $v$ can receive at most one charge in relation to arc $A_2$.
This implies that $v$ can receive at most one more charge in total, which  corresponds to
$A_3$.
\end{proof}

Figure~\ref{fig:cdt-steps} illustrates the incremental construction
for an indicated order $o=(v_1,\ldots, v_6)$. Vertices  $v_4$ and $v_6$ coincide.
The insertion of $v_5$ causes the arc of $v_4$ to
split, see Fig.~\ref{fig:cdt-steps}(c). The result of deleting the
created auxiliary arc is shown in Fig.~\ref{fig:cdt-steps}(d);
we insert $v_6$ 
in Fig.~\ref{fig:cdt-steps}(e).
In this example, we could avoid deleting the
auxiliary arc of $v_4$, which is created by inserting  $v_5$ in  Fig.~\ref{fig:cdt-steps}(c), because it overlaps with an
arc of $C$, therefore, it is known that it will later be
re-inserted and it cannot obstruct the insertion process of other arcs.

     \begin{figure}
	\centering 
	\includegraphics[scale=0.7]{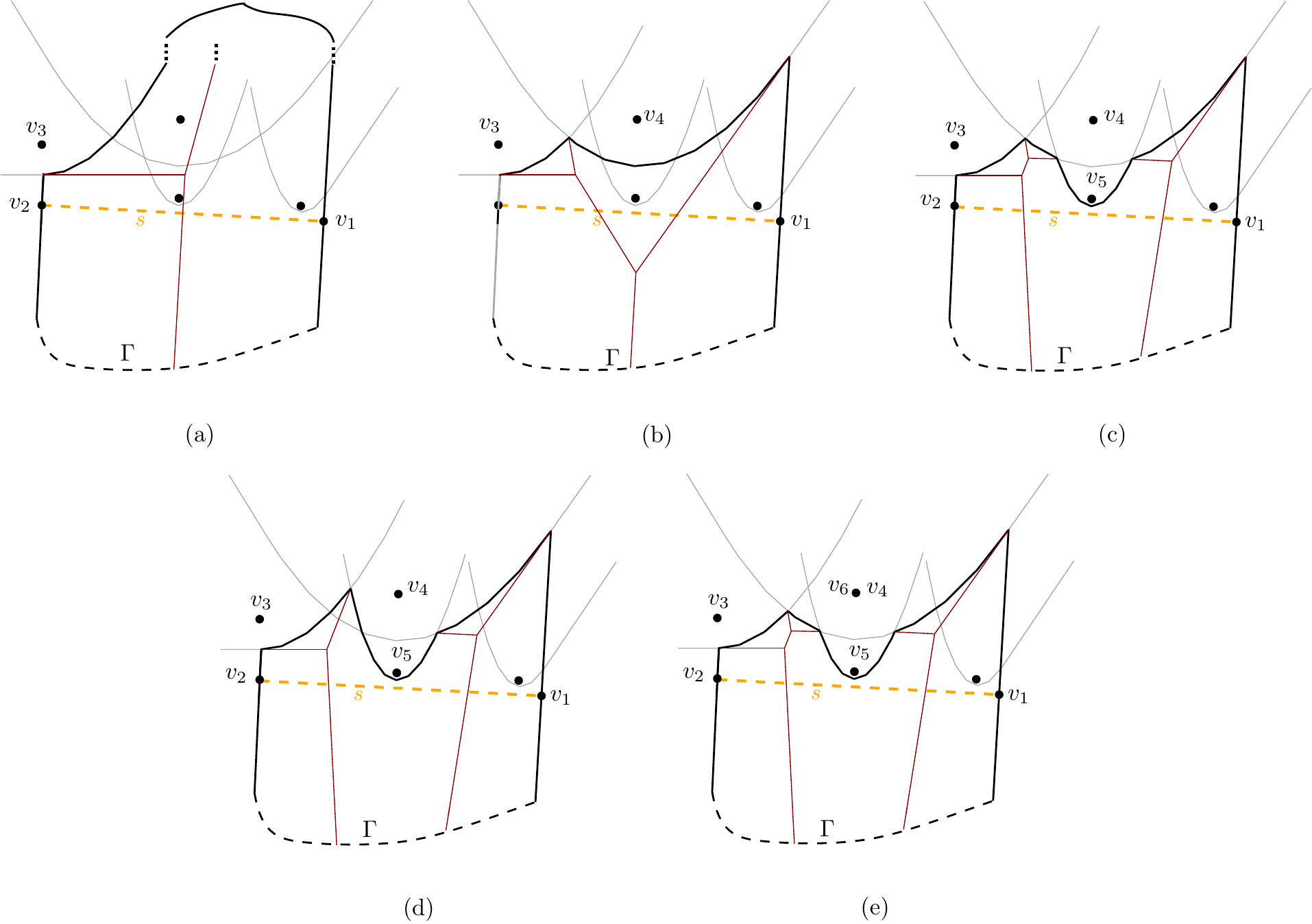}
        \caption{The incremental construction 
          following $o=(v_1,\ldots, v_6)$.
          (a) $\vld(C_3)$; (b) $\vld(C_4)$; (c)  $\vld(C_5')$;
          (d) $\vld(C_5)$; (e) $\vld(C_6)$.}
	\label{fig:cdt-steps}
      \end{figure}

      \section{Concluding remarks}
We have also considered the variant of computing, in linear
expected time, a Voronoi-like 
tree (or forest) within a simply connected domain $D$, of constant 
boundary complexity, given the ordering of some Voronoi faces along the 
boundary of $D$. In an extended paper, we will provide conditions under which the same
essentially technique can be applied. 

In future research, we are also interested in 
considering
deterministic linear-time algorithms to compute
abstract Voronoi-like trees and forests as inspired by \cite{AGSS89}.

\bibliography{voronoi}

\begin{thebibliography}{10}

\bibitem{AGSS89}
Alok Aggarwal, Leonidas~J. Guibas, James~B. Saxe, and Peter~W. Shor.
\newblock A linear-time algorithm for computing the {V}oronoi diagram of a
  convex polygon.
\newblock {\em Discrete and Computational Geometry}, 4:591--604, 1989.

\bibitem{aurenbook}
Franz Aurenhammer, Rolf Klein, and Der-Tsai Lee.
\newblock {\em Voronoi Diagrams and Delaunay Triangulations}.
\newblock World Scientific, 2013.

\bibitem{BKL17}
Cecilia Bohler, Rolf Klein, and Chih-Hung Liu.
\newblock Abstract {V}oronoi diagrams from closed bisecting curves.
\newblock {\em International Journal of Computational Geometry and
  Applications}, 2017.

\bibitem{BKL19}
Cecilia Bohler, Rolf Klein, and Chih-Hung Liu.
\newblock An efficient randomized algorithm for higher-order abstract voronoi
  diagrams.
\newblock {\em Algorithmica}, 81(6):2317--2345, 2019.

\bibitem{Chew90}
L.~Paul Chew.
\newblock Building {Voronoi} diagrams for convex polygons in linear expected
  time.
\newblock Technical report, Dartmouth College, Hanover, USA, 1990.

\bibitem{Del34}
B.~Delaunay.
\newblock Sur la sph\`ere vide. {A} la memoire de {G}eorges {V}oronoi\"\i.
\newblock {\em Bulletin de l'Acad\'emie des Sciences de l'URSS}, 6:793–800,
  1934.

\bibitem{ES86}
H.~Edelsbrunner and R.~Seidel.
\newblock Voronoi diagrams and arrangements.
\newblock {\em Discrete and Comptational Geometry}, 1:25--44, 1986.

\bibitem{JP18}
Kolja Junginger and Evanthia Papadopoulou.
\newblock {Deletion in Abstract Voronoi Diagrams in Expected Linear Time}.
\newblock In {\em 34th International Symposium on Computational Geometry
  (SoCG)}, volume~99 of {\em LIPIcs}, pages 50:1--50:14, Dagstuhl, Germany,
  2018.

\bibitem{JP19}
Kolja Junginger and Evanthia Papadopoulou.
\newblock On tree-like abstract {V}oronoi diagrams in expected linear time.
\newblock In {\em CGWeek Young Researchers Forum (CG:YRF)}, 2019.

\bibitem{JP20arxiv}
Kolja Junginger and Evanthia Papadopoulou.
\newblock Deletion in abstract {V}oronoi diagrams in expected linear time and
  related problems.
\newblock arXiv:1803.05372v2 [cs.CG], 2020.
\newblock To appear in \emph{Discrete and Computational Geometry}.

\bibitem{K89}
Rolf Klein.
\newblock {\em Concrete and Abstract {Voronoi} Diagrams}, volume 400 of {\em
  Lecture Notes in Computer Science}.
\newblock Springer-Verlag, 1989.

\bibitem{KLN09}
Rolf Klein, Elmar Langetepe, and Z.~Nilforoushan.
\newblock Abstract {V}oronoi diagrams revisited.
\newblock {\em Computational Geometry: Theory and Applications},
  42(9):885--902, 2009.

\bibitem{KMM93}
Rolf Klein, Kurt Mehlhorn, and Stefan Meiser.
\newblock Randomized incremental construction of abstract {Voronoi} diagrams.
\newblock {\em Computational Geometry: Theory and Applications}, 3:157--184,
  1993.

\bibitem{sharir95}
Micha Sharir and Pankaj~K. Agarwal.
\newblock {\em Davenport-Schinzel sequences and their geometric applications}.
\newblock Cambridge university press, 1995.

\bibitem{SB15}
Jonathan~Richard Shewchuk and Brielin~C. Brown.
\newblock Fast segment insertion and incremental construction of constrained
  {D}elaunay triangulations.
\newblock {\em Computational Geometry: Theory and Applications},
  48(8):554--574, 2015.

\end{thebibliography}

\end{document}